\documentclass[sigconf, nonacm]{acmart}

\newcommand\vldbdoi{XX.XX/XXX.XX}
\newcommand\vldbpages{XXX-XXX}
\newcommand\vldbvolume{20}
\newcommand\vldbissue{1}
\newcommand\vldbyear{2027}
\newcommand\vldbauthors{\authors}
\newcommand\vldbtitle{\shorttitle}
\newcommand\vldbavailabilityurl{}
\newcommand\vldbpagestyle{plain}

\usepackage{booktabs}       % Professional tables
\usepackage{tikz-cd}        % Commutative diagrams
\usepackage{listings}       % Code / PDL / verification-query listings
\usepackage{amsmath}
\usepackage{enumitem}       % itemize/enumerate spacing control
\usepackage{array}          % fixed-width column types
\newcolumntype{L}[1]{>{\raggedright\arraybackslash}p{#1}} % ragged-right p-column
\lstdefinestyle{pda}{%
  basicstyle=\ttfamily\footnotesize,
  columns=fullflexible,
  keepspaces=true,
  morecomment=[l]{--},        % Agda/PDA line comments
  commentstyle=\itshape,
  frame=single,
  framesep=4pt,
  xleftmargin=4pt,
  aboveskip=2pt, belowskip=0pt,
}

\newif\ifextended
\extendedtrue
\newcommand{\fullproof}{\ifextended Full proof in Appendix~\ref{app:proofs}.\else
Full proof in the extended version of this paper.\fi}

\newif\ifresults
\resultsfalse

\newif\ifarxiv
\arxivtrue

\newif\ifbranded
\brandedfalse
\ifbranded
  \input{branding-private}
\else
  \newcommand{\sysname}{PDD-Toolchain}       % software system (proper noun)
  \newcommand{\Sysname}{PDD-Toolchain}       % sentence-start form (same)
  \newcommand{\edbname}{EDB}                 % executable-design framework
  \newcommand{\edbfull}{Executable Design Builder}
  \newcommand{\pddskill}{\texttt{/pdd}}      % design + verification skill
  \newcommand{\buildskill}{\texttt{/build-pip}}
  \newcommand{\ddlskill}{\texttt{/build-ddl}}
  \newcommand{\jobskill}{\texttt{/build-jobs}}
  \newcommand{\brandednote}{}                % branded-only aside (empty here)
\fi

\newcommand{\grain}[1]{G[#1]}              % G[R] — grain operator
\newcommand{\ek}[1]{EK[#1]}                % EK[R] — entity-key operator
\newcommand{\BC}[1]{BC[#1]}                % BC[R] — behavioral-class operator
\newcommand{\eqg}{\equiv_g}                % grain equality
\newcommand{\leg}{\leq_g}                  % grain ordering
\newcommand{\incg}{\langle\rangle_g}       % grain incomparability
\newcommand{\gl}[1]{\varphi(#1)}           % φ(h) — grain lift
\newcommand{\subt}{\subseteq_{typ}}        % type subset
\newcommand{\psub}{\subset_{typ}}          % proper type subset
\newcommand{\tun}{\cup_{typ}}              % type union
\newcommand{\tin}{\cap_{typ}}              % type intersection
\newcommand{\tdiff}{-_{typ}}               % type difference
\newcommand{\thra}{\twoheadrightarrow}     % surjection arrow

\newcommand{\Coll}[1]{C\,#1}               % C R — data collection
\newcommand{\calcg}{\textsf{CalcG}}        % grain calculation operator
\newcommand{\seqc}{\mathbin{>\!>}}          % >>  forward composition
\newcommand{\parc}{\mathbin{<\!\times\!>}} % <×> pipeline product
\newcommand{\forkc}{\mathbin{<\!\&\&\!>}}  % <&&> pipeline fork (diagonal)
\newcommand{\seqp}{\mathbin{>\!>\!p}}       % >>p pipeline sequencing
\newcommand{\cta}[3]{[\,#1 \Rightarrow #2 \Rightarrow #3\,]} % CTA constructor
\newcommand{\pda}{\textsf{PDA}}
\newcommand{\pdd}{\textsf{PDD}}
\newcommand{\sem}[1]{[\![#1]\!]}            % denotation brackets
\newcommand{\grainproj}[1]{\textit{grain}_{#1}} % grain projection R -> G[R]

\theoremstyle{plain}
\newtheorem{theorem}{Theorem}[section]

\newtheorem{corollary}[theorem]{Corollary}
\newtheorem{proposition}[theorem]{Proposition}
\theoremstyle{definition}

\theoremstyle{remark}
\newtheorem{remark}[theorem]{Remark}

\begin{document}

%% --- Title -----------------------------------------------------------------
% Methodology-first: title names the contribution (PDD), never the software system.
% The software realization appears only in the realization (\S7) and evaluation (\S8),
% and only under \brandedtrue. Running head = the methodology.
\title[Pipeline Denotational Design]{Pipeline Denotational Design: Correct-by-Construction Data Pipelines at Zero Cost}

%% Single-blind: names and affiliations REQUIRED on page 1.
\author{Nikos Karayannidis}
\affiliation{%
  \institution{Independent}
  \city{Athens}
  \country{Greece}
}
\email{nkarag@gmail.com}

%% --- Abstract --------------------------------------------------------------
\begin{abstract}
\emph{Pipeline Denotational Design} (PDD) is a design-first methodology for building
data pipelines that are correct by construction. As AI agents generate pipeline code
at scale, the bottleneck shifts from \emph{writing} pipelines to \emph{verifying}
them---and the errors that matter most, such as grain inconsistencies that silently
inflate aggregates, are invisible to schema checks, type checks, and tests on sampled
data. PDD designs pipelines in a \emph{semantic domain} rather than in code: a design
is composed from a typed algebra of operations---the Pipeline Design Algebra (PDA), one
such instantiation---in which every well-typed composition is grain-correct by
construction. Resting on grain alone, this guarantee is \emph{universal}: it holds for
any grain-inferring operation set, over any engine, batch or streaming (PDA being our
choice). Correctness is established in three layers---grain, behavioral class, and
domain---all \emph{at design time, at zero cost}, with no access to data: grain by a
data-independent type-level computation (CalcG), behavioral class by the type checker,
and the domain rules by a proof-carrying composition over operation contracts. This
rests on a theorem (\emph{Pipeline Correctness}): the three layers discharge grain,
behavioral class, and every contract-derivable obligation with no data access, and a
proof-carrying composition collapses what remains to a \emph{single input-boundary
check}---so a pipeline is correct \emph{by construction against its specification}, and
the only data-dependent residue is whether real inputs meet the design's preconditions
(\emph{data quality}, not code correctness). PDD \emph{generates} that boundary check as
SQL/PySpark verification queries, provably faithful to their specification. Correctness by
construction is a \emph{spectrum}, not a single fixed point: the same design can be
verified at three levels of rigor---by running checks against the data at runtime, by
the deployed type-level checker (which decides grain and behavioral class automatically,
needing no external solver), or by machine-checked proofs in Agda/Lean~4. The stronger
the guarantee, the less is left to check at runtime, and a team chooses where on this
spectrum to sit by a pipeline's criticality, complexity, and scale. This recasts the
engineer's role: an AI agent instantiates a pre-verified pattern and ships a
machine-checkable correctness \emph{certificate}---an Agda proof or an evidence
ledger---so the human validates only a compact, data-anchored specification and checks
the certificate, a \emph{proof-carrying} discipline for AI-generated pipelines. We
realize the methodology end-to-end in a production toolchain---a
design tool, a typed framework that implements the algebra, and code-generation that
compiles a verified design into a faithful implementation---and
\ifresults
evaluate it on production pipelines spanning multiple patterns, behavioral classes,
and modeling paradigms, measuring design-time defect capture, the zero-cost discharge
ratio, generated-query yield, and code-generation faithfulness against ad-hoc,
runtime-data-quality, and LLM-without-design baselines.
\else
set out an evaluation on production pipelines spanning multiple patterns, behavioral
classes, and modeling paradigms---a three-arm protocol targeting design-time defect
capture, the zero-cost discharge ratio, generated-query yield, and code-generation
faithfulness against ad-hoc, runtime-data-quality, and LLM-without-design baselines
(the measurement campaign is in progress).
\fi Finally, because the guarantee
rests on grain alone, the same design-time check reaches beyond pipelines to the
AI-generated \emph{queries} run over semantic layers and ontologies---a correctness
layer that stack otherwise lacks.
\end{abstract}

\maketitle

\ifarxiv
% arXiv preprint: clean front matter, no PVLDB reference/DOI/copyright block.
\pagestyle{plain}
\else
%%% do not modify the following VLDB block %%
%%% VLDB block start %%%
\pagestyle{\vldbpagestyle}
\begingroup\small\noindent\raggedright\textbf{PVLDB Reference Format:}\\
\vldbauthors. \vldbtitle. PVLDB, \vldbvolume(\vldbissue): \vldbpages, \vldbyear.\\
\href{https://doi.org/\vldbdoi}{doi:\vldbdoi}
\endgroup
\begingroup
\renewcommand\thefootnote{}\footnote{\noindent
This work is licensed under the Creative Commons BY-NC-ND 4.0 International License. Visit \url{https://creativecommons.org/licenses/by-nc-nd/4.0/} to view a copy of this license. For any use beyond those covered by this license, obtain permission by emailing \href{mailto:info@vldb.org}{info@vldb.org}. Copyright is held by the owner/author(s). Publication rights licensed to the VLDB Endowment. \\
\raggedright Proceedings of the VLDB Endowment, Vol. \vldbvolume, No. \vldbissue\ %
ISSN 2150-8097. \\
\href{https://doi.org/\vldbdoi}{doi:\vldbdoi} \\
}\addtocounter{footnote}{-1}\endgroup
%%% VLDB block end %%%
\fi

%%% do not modify the following VLDB block %%
%%% VLDB block start %%%
\ifdefempty{\vldbavailabilityurl}{}{
\vspace{.3cm}
\begingroup\small\noindent\raggedright\textbf{PVLDB Artifact Availability:}\\
The source code, data, and/or other artifacts have been made available at \url{\vldbavailabilityurl}.
\endgroup
}
%%% VLDB block end %%%

%% --- Body ------------------------------------------------------------------
% =============================================================================
% 1. INTRODUCTION   (budget ~1.25 pp)
% Sources: roadmap §4 (framing); exec-deck slide 3 (enterprise-complexity problem
% framing, anonymized); arXiv intro (fan-trap hook); PDD.md §II.0.
% =============================================================================

\section{Introduction}
\label{sec:intro}

Data pipelines at enterprise scale are staggeringly complex. A single analytics
platform integrates dozens of heterogeneous sources---operational databases, APIs,
message queues, blob storage---through thousands of interdependent transformation
units, across several persistence layers (landing, integration, serving---the medallion
bronze/silver/gold layers), in a mix of
cadences from streaming and micro-batch to periodic batch. This complexity is hard to
understand, reason about, debug, and evolve; and \emph{verifying that a pipeline is
correct} at this scale has always been a struggle. Tests help but do not suffice---as
Dijkstra put it, ``program testing can be used to show the presence of bugs, but never
to show their absence''---and the bugs that matter most slip through precisely the
tests engineers rely on. Worse, AI agents now generate pipeline code at scale, faster
than anyone can review it, amplifying inconsistency and error. The problem this paper
addresses is therefore twofold: \emph{tame} this complexity, and \emph{verify}
correctness with guarantees stronger than testing can give.

The errors that matter are \emph{silent}. Consider a data engineer---or, increasingly,
an AI agent---building a report that combines revenue from one fact table with units
sold from another, aggregated per customer and date. The two facts have different
\emph{grains}: one row per
customer-channel-date, the other per customer-product-date. The natural query
joins on the common columns and aggregates:
\begin{flushleft}\small\ttfamily
SELECT customer\_id, date, SUM(revenue), SUM(units)\\
FROM sales\_channel JOIN sales\_product USING (customer\_id, date)\\
GROUP BY customer\_id, date
\end{flushleft}
It compiles, runs, and returns the expected result grain---yet it
\emph{systematically inflates both metrics}. The join creates a cross product over
the unmatched grain components (channels against products), duplicating every row
before aggregation. This is a \emph{fan trap}, and it is invisible to schema
validation, type checking, and unit tests on small data: with one channel and one
product per customer-date the cross product is $1\times1$ and the numbers look
right. The root cause is a \emph{grain inconsistency}---an unintended change in the
level of detail at which data is represented---and it surfaces only in production,
after the data has been processed, loaded, and consumed downstream.

\paragraph{Who verifies AI-generated pipelines?} Errors of this kind are about to
get worse. AI agents now generate pipeline code at scale, from natural-language
prompts, faster than any human can review it. The agents are fluent at producing
code that runs; they are no better than the engineer at noticing that the result
is silently wrong. The question the field must answer is not ``can an agent write a
pipeline?''---it plainly can---but \emph{who verifies that the pipeline is
correct?} Testing does not scale to the volume of generated code, and---as the fan
trap shows---tests on sampled data do not even catch the bugs that matter. The remedy
is not better detection---the data that triggers such a bug and the verification query
that would catch it seldom coincide---but \emph{prevention}: pipelines designed so the
bug cannot be written, in the discipline of making illegal states
unrepresentable~\cite{minsky2011ocaml} (\S\ref{sec:correctness:prevention}). The same
question is spreading to the read path: agents increasingly generate \emph{queries}
too, over semantic layers and ontologies that say \emph{where} and \emph{how} to ask
but not \emph{whether} the answer is correct---so the same design-time guarantee is
what that stack is missing (\S\ref{sec:beyond}).

\paragraph{From verification to synthesis.} This paper is \emph{Part~II} of a
two-part program. \emph{Part~I}~\cite{graintheory-pods} established the theory and
\emph{verification}: it formally defines the denotation of data---grain, entity key,
behavioral class---with computable grain-inference rules and proves the \emph{grain
homomorphism} on which everything here rests, so that given a pipeline annotated with
grains, the CalcG algorithm decides whether the output grain matches the target's
declared grain---a data-independent, decidable, zero-cost check. Verification answers
``is \emph{this} pipeline correct?''. This paper takes the next step---\emph{synthesis}
and \emph{methodology}---and answers ``how do I \emph{build} a pipeline that cannot be
incorrect?''. We present \emph{Pipeline
Denotational Design} (PDD): a methodology in which pipelines are composed from a
typed algebra of operations (the \emph{Pipeline Design Algebra}, PDA) where
\emph{every well-typed composition is grain-correct by construction}, and the
residual, genuinely data-dependent obligations are discharged as
\emph{automatically generated} verification queries. The algebra it instantiates,
\pda{}, is to data
pipelines what
a typed query builder (e.g.\ LINQ, jOOQ, Opaleye) is to raw SQL, and what Codd's relational
algebra~\cite{codd1970relational} is to pre-relational data wrangling: a small algebra
with laws in which correctness is a \emph{structural} property---settled by
construction---rather than a test result. The two comparisons stress complementary
facets: \emph{static type-safety} (a malformed pipeline is rejected before it runs) and
a \emph{foundational algebra with equational laws} (compose and reason, rather than
hand-code)---and \pda{} claims both at once, one level up: over whole pipelines, not single
queries.

\paragraph{Three levers on complexity.} PDD attacks the twofold problem on three
fronts. \emph{Simplicity by abstraction}: designers reason in the semantic domain
(grain, entity key, behavioral class), formally but away from engine-specific code, so
the centre of gravity moves from implementation to meaning. \emph{Uniformity by
patterns}: a small, typed vocabulary of operations and pre-verified patterns makes
every pipeline look the same modulo its data domain (a pattern is a pipeline
parameterized by that domain, \S\ref{sec:synthesis}), turning bespoke plumbing into
reusable structure. \emph{The design as a contract with AI}: the denotational design
is a formal blueprint an agent must satisfy and provide design-time evidence for, so AI
accelerates delivery without amplifying error. Simplicity and uniformity tame the
complexity; the contract, together with the three correctness layers below, verifies
the pipeline's correctness.

\paragraph{Correctness in three layers, at design time, at zero cost.} PDD verifies a
pipeline in three layers (\S\ref{sec:correctness}): \emph{grain} (does the output
grain match the target?), \emph{behavioral class} (does each operation respect the
read/write semantics of entities, events, versioned records, \dots?), and
\emph{domain} (do the business rules hold?). All three are verified \emph{at design
time, at zero cost}---with no access to data: grain by a type-level computation,
behavioral class by the type checker, and the domain rules by a proof-carrying
composition that establishes them on the output \emph{given the design's input
preconditions}. A data engineer controls the pipeline he builds, not the data a
provider sends through it; so the one thing left for runtime is to check that those
inputs meet the preconditions---\emph{data quality}, not code correctness---which the
method discharges as auto-generated, provably faithful validation queries. The slogan
``provably correct pipelines at zero cost, by construction'' is thus literal: a
pipeline's correctness is settled before any data is read.

\paragraph{Three concerns, cleanly separated.} Runtime testing conflates three
questions that PDD keeps apart. \emph{(i)~Design correctness}---does the pipeline
meet its specification?---is a property of the \emph{code}, settled at design time by
a type system (the \emph{Pipeline Correctness Theorem}, \S\ref{sec:correctness}): as
in Agda or Lean~4, business rules become types and a proof-carrying composition
propagates pre/postconditions so the design type-checks \emph{against its stated
target contract} before any data flows. \emph{(ii)~Implementation faithfulness}---does
the executable compute what the design denotes?---is discharged not by data but by the
grain homomorphism (\S\ref{sec:background}) and compilation-correctness
(Thm.~\ref{thm:compile}), which carry the design's meaning to the running engine.
\emph{(iii)~Data quality}---do real inputs satisfy the design's assumed
preconditions?---is the one genuinely data-dependent question, the boundary every
typed system keeps against untyped external data. Design-time checking eliminates
\emph{code} bugs (grain mismatches, BC violations, spurious joins); the homomorphism
governs the design-to-executable gap; and the generated queries guard against
\emph{data} bugs at that boundary. The
checker is a \emph{role}. At its dependent-types limit \emph{`mostly at design time'}
becomes \emph{`fully at design time'}: in PDL---the \emph{Pipeline Design Language},
our embedded DSL for pipeline design in Agda and Lean~4, encoding \pda{} and its
constraint logic---even the boundary constraints become types, so a design that type-checks is
\emph{completely} correct, carrying a machine-checked proof a reviewer discharges by recompiling; the
runtime residue then shrinks to the one check no type system can absorb---that
\emph{untyped} incoming data meets the design's preconditions. \Sysname{} as
deployed sits one notch left of that limit: the \pddskill{} tool plays the checker
role with an evidence document over the grain, behavioral-class, and
contract-derivable obligations---lighter than a proof assistant, but enough to make
the principle practical in production. \S\ref{sec:correctness:spectrum} places the
two on one spectrum. This is the data-pipeline instance of software engineering's
tests-vs-types-vs-proofs spectrum: with only schemas (structure) and untyped
operations, a pipeline's sole handle on correctness is to \emph{run} it and test
outputs---the data tests the field relies on today, which on sampled data miss the
grain bugs that matter. A type system rich enough to carry data \emph{with its
denotation} and operations as contracts propagated through composition
\cite{graintheory-pods} moves \emph{code} correctness to design time and confines
runtime checks to \emph{data quality} alone.

\paragraph{The engineer's new role: author the spec, check the certificate.}
Prevention reshapes \emph{who does what}. Formal verification is only ever as good as
its specification---a proof discharges the properties you state and is silent on the one
you forgot---and PDD turns that classical limitation into a division of labor. Its
specification is not a natural-language wish but the data \emph{denotation}
$(\grain{R},\ek{R},\BC{R})$ of source and target, plus the pre/postcondition business
rules (PDLC Step~1, \S\ref{sec:overview:pdlc}); and much of it is not the engineer's to
author---the \emph{source} denotation is \emph{inferred} from data during the
data-source analysis any pipeline already does, PDD or not, and the grain and
behavioral-class obligations are \emph{derived}---so the properties most easily forgotten
(don't double-count, respect versioning) cannot be omitted, and what remains his is
small: the target model and its business predicates. Given that spec, the agent does not
free-style: it faithfully instantiates the \emph{designated} pipeline pattern for that layer---each
layer of the architecture has one or more pre-verified patterns (\S\ref{sec:synthesis})---%
carrying the data from the source denotation to the target's without violating a
constraint, and emits a machine-checkable \emph{certificate} of correctness---an Agda
proof at the PDL end, an evidence ledger at the deployed end. The trust model is thus
\emph{proof-carrying}~\cite{necula1997pcc}: the agent is an untrusted producer that must
ship a certificate; the human is the trusted consumer who does only two
things---\emph{validate the spec} (the irreducible act no verifier can perform) and
\emph{check the certificate} (recompile the proof, or review the ledger). Because proofs
are costly to produce but cheap to check, this is exactly what makes delegating synthesis
to an AI safe. The engineer's role shifts from writing and reviewing code---infeasible at
the scale AI now generates it---to authoring a compact, data-anchored
\emph{specification} and checking the certificate rather than the code. Whether the agent
reliably \emph{reaches} a faithful instantiation is a separate, empirical question,
which our evaluation answers by comparing pipelines built with and without PDD
(\S\ref{sec:eval}).

\paragraph{A production realization.} We realize the methodology end-to-end in \sysname{}
(\S\ref{sec:toolchain}): a design tool that produces and verifies a formal design,
a typed PySpark framework (\edbname) implementing the PDA operations on a real engine,
and code-generation that compiles a verified design into a faithful implementation
plus its verification queries. \Sysname{} is in production use on an enterprise
data-engineering program---not a proof of concept---and is the source of our
evaluation (\S\ref{sec:eval}).

\paragraph{Contributions.}
\begin{itemize}[leftmargin=1.2em,itemsep=2pt]
\item \textbf{Pipeline Denotational Design} (PDD, \S\ref{sec:overview}): a
design-first methodology that settles a pipeline's correctness \emph{at design time,
at zero cost}, by designing in a semantic domain and deriving the implementation to
match---grounded in the grain homomorphism (\S\ref{sec:background}) that makes
design-level reasoning provably faithful to the data.
\item A \textbf{three-layer correctness framework} (\S\ref{sec:correctness})---grain
(CalcG), behavioral class (typing), and domain (contract-derivable at zero cost;
data-dependent as generated, provably-faithful queries, Thm.~\ref{thm:compile})---%
unified by the \textbf{Pipeline Correctness Theorem} (Thm.~\ref{thm:pct}): the
structural layers are discharged with no data access, and the proof-carrying
composition collapses the residual data-dependent obligations to a \emph{single
input-boundary check}, so the design is correct by construction \emph{against its
specification}---the only question left for data is whether real inputs meet that leaf
precondition. It delivers \emph{prevention, not detection}---grain, behavioral-class, and
fan/chasm-trap bugs become unrepresentable (\S\ref{sec:correctness:prevention}).
\item A \textbf{verification spectrum} (\S\ref{sec:correctness:spectrum}): the same PDD
design can be checked at three levels of rigor---runtime queries against the data, the
deployed type-level checker (which decides grain and behavioral class automatically, with
no external solver), or machine-checked Agda/Lean~4 proofs. The stronger the check, the
less is left to verify at runtime---from a check at every node down to a single check at
the input. A team picks its level by a pipeline's criticality, complexity, and scale:
correctness by construction is a dial, not a fixed cost.
\item \textbf{PDD Universality} (Thm.~\ref{thm:opgen}): the guarantee rests on grain
alone, so \emph{any} grain-inferring operation set---over batch or streaming
carriers---instantiates PDD, with worked instantiations for PDA, relational algebra,
and the dataflow model.
\item The \textbf{Pipeline Design Algebra} (PDA, \S\ref{sec:pda}), the concrete
\emph{instrument} we instantiate: four morphism classes (Capture/Transform/Declare/%
Apply) and a $\sim$30-operation catalog, each carrying a contract fixing how it
transforms grain, behavioral class, and cardinality; composition operators and
pre-verified patterns for which well-typed composition is grain-correct by
construction (Thm.~\ref{thm:synthesis}), plus a faithful bridge from design to
executable code.
\item \textbf{PDL}, the Pipeline Design Language (\S\ref{sec:related}): a
dependently-typed embedding of PDA and its
constraint logic in Agda/Lean~4 with by-construction pipeline-correctness proofs---the
proof end of the tests--types--proofs spectrum, mechanized end-to-end on a real
pipeline.
\item A \textbf{production realization}, \sysname{} (\S\ref{sec:toolchain}): a design
tool, a typed framework implementing the PDA operations on a real engine, and
code-generation that compiles a verified design into a faithful implementation---in
production use on an enterprise program, not a proof of concept.
\item \ifresults An \textbf{experimental evaluation}\else An \textbf{evaluation
design}\fi{} on real production pipelines
(\S\ref{sec:eval})---design-time defect capture, the zero-cost discharge ratio,
generated-query yield, code-generation faithfulness, cross-pipeline uniformity, and
runtime residue-collapse---that builds each pipeline three ways along the verification
spectrum: \emph{without} PDD (an agent codes directly from the specification, with no
formal design), \emph{with} PDD checked by the deployed evidence-emitting checker, and
\emph{with} PDD proved in PDL (Agda/Lean~4). The contrast isolates the design step from
the LLM's coding ability (no-PDD vs.\ deployed) and measures what the proof end adds
(proved vs.\ deployed). Ad-hoc hand-written pipelines and runtime-only data-quality
tooling serve as reference points that locate the arms on the wider landscape---the
latter notably \emph{unable to express} the grain and behavioral-class defects at all.
\end{itemize}

\noindent Together they make correctness-by-construction for data pipelines a
practical \emph{methodology}: grounded in theory (the grain homomorphism), broad
(three layers, not grain alone), realized in production, and evaluated quantitatively
on real pipelines against agentic and runtime-data-quality baselines.

% =============================================================================
% 2. OVERVIEW: WHAT PIPELINE DENOTATIONAL DESIGN IS   (budget ~1.0 pp)
% Sources: PDD.md §II.0 (semantic domain, lines 718–747); roadmap §4 (Codd
% analogy, three-layer framework). First full draft.
% =============================================================================

\section{Pipeline Denotational Design}
\label{sec:overview}

PDD adapts \emph{denotational design}~\cite{elliott2009denotational}---specifying a
system by a precise mathematical meaning and deriving the implementation to match
it---to data pipelines. Rather than designing in terms of SQL queries, PySpark
jobs, or Airflow DAGs, the engineer designs in a \emph{semantic domain} and lets
the implementation be generated to agree with that design; we call this the
\emph{executable design}.

\subsection{The Semantic Function}

The bridge between the two worlds is a semantic function $\sem{\cdot}$, a
structure-preserving map from concrete artifacts to their meanings
(Table~\ref{tab:semantic-function}). A table, DataFrame, or file denotes a
\emph{data collection} $c:\Coll{R}$ (\S\ref{sec:pda:collection}) carrying a grain
$\grain{R}$, entity key $\ek{R}$, and behavioral class $\BC{R}$; a SQL/PySpark/dbt
transformation denotes a
declarative \emph{PDA operation} (\S\ref{sec:pda}); an entire ETL pipeline
denotes a \emph{function} $\textsf{Pipeline}\;\textsf{Source}\;\textsf{Target}$
(\S\ref{sec:synthesis}); and a business rule or integrity requirement denotes a
\emph{data constraint}---a first-order-logic predicate (\S\ref{sec:correctness}). Following the discipline of denotational
design~\cite{elliott2009denotational}, $\sem{\cdot}$ is a single \emph{meaning}
function overloaded by the \emph{category} of its argument---written
$\sem{\cdot}_{\mathcal{C}}$ when the category $\mathcal{C}$ is ambiguous and dropped
otherwise; for each category it fixes a model $\sem{\mathcal{C}}$ and a map
$\sem{\cdot}_{\mathcal{C}}:\mathcal{C}\to\sem{\mathcal{C}}$. Its most fundamental case
is the \emph{data type} itself: a type $R$ denotes the \emph{denotation of data}
$\sem{R}=(\grain{R},\ek{R},\BC{R})$ of Part~I~\cite{graintheory-pods}, the triple every
collection carries. The defining property of $\sem{\cdot}$ is that it is a
\emph{homomorphism}:
\begin{equation}
\label{eq:homomorphism}
\sem{p\;\textsf{src}} \;=\; \sem{p}\;\sem{\textsf{src}}.
\end{equation}
The meaning of running a pipeline on data equals the meaning of the pipeline
applied to the meaning of the data. Equation~\eqref{eq:homomorphism} is what lets
design and implementation be developed independently and stay synchronized:
reasoning carried out on the design (the right-hand side) is faithful to what the
code actually computes (the left-hand side). \S\ref{sec:background} proves that
this homomorphism holds---it is the \emph{grain homomorphism} of grain theory---and
explains why it is \emph{computable}, which is what turns the principle into a
practical method.

\begin{table}[t]
\centering\small
\caption{The meaning function $\sem{\cdot}$, overloaded by argument category: for each
category $\mathcal{C}$ it gives a model $\sem{\mathcal{C}}$ and a semantic function
$\sem{\cdot}_{\mathcal{C}}:\mathcal{C}\to\sem{\mathcal{C}}$ (after
Elliott~\cite{elliott2009denotational}). The data-type row is the foundational case
(Part~I~\cite{graintheory-pods}); the next rows map the concrete artifacts that realize
it, and the last gives a \emph{data constraint} (a business rule or integrity
requirement) its denotation---a FOL predicate, realized concretely as a verification
query.}
\label{tab:semantic-function}
{\setlength{\tabcolsep}{3pt}
\begin{tabular}{@{}l c L{4.2cm}@{}}
\toprule
Argument & $\sem{\cdot}$ & Meaning (PDD) \\
\midrule
data type $R$ & $\to$ & $(\grain{R},\ek{R},\BC{R})$: denotation of data~\cite{graintheory-pods} \\
table / DataFrame / file & $\to$ & collection $\Coll{R}$ carrying $(\grain{R},\ek{R},\BC{R})$ \\
SQL / PySpark / dbt step & $\to$ & PDA operation (declarative intent) \\
ETL pipeline & $\to$ & $\textsf{Pipeline}\;\textsf{Source}\;\textsf{Target}$ (a function) \\
data constraint & $\to$ & FOL predicate: a $(\textsf{Pre},\textsf{Intra},\textsf{Post})$ spec \\
\bottomrule
\end{tabular}}
\end{table}

\subsection{Four Denotations, One Methodology}

PDD assigns meaning at four levels, each the subject of a later section: the
\emph{denotation of data} (\S\ref{sec:background})---what data \emph{is}: grain,
entity key, behavioral class; the \emph{denotation of operations}
(\S\ref{sec:pda})---what a transformation \emph{intends}: a PDA operation with a
contract; the \emph{denotation of pipelines} (\S\ref{sec:synthesis})---what an
end-to-end workflow \emph{means}: a composition of those operations; and the
\emph{denotation of data constraints} (\S\ref{sec:correctness})---what a
\emph{requirement} means: a first-order-logic predicate. Working in
this domain makes designs \emph{precise} (formal types remove ambiguity),
\emph{composable} (the meaning of the whole is the composition of the meanings of
the parts), \emph{verifiable} (correctness conditions are decidable or provable),
and \emph{implementation-independent} (a design is a specification, not code tied
to one engine).

The fourth denotation is worth a closer look, as it carries the domain layer of
\S\ref{sec:correctness}. A business rule, integrity constraint, or data-integration
rule denotes a first-order-logic predicate, realized in the type system as a static
relation between input and output ($\textsf{ConsStatic}$) or a stateful one over a
target's before/after state ($\textsf{ConsDynamic}$), and bundled with its
precondition into a $(\textsf{Pre}$, $\textsf{Intra}$, $\textsf{Post})$ specification. Such a
specification attaches to a pipeline or sub-pipeline (a pipeline is a DAG of pipelines)
and \emph{propagates through the DAG} via the pipeline operators
(\S\ref{sec:synthesis}), each composition discharging the obligation that one step's
$\textsf{Post}$ entails the next step's $\textsf{Pre}$. Its concrete realization is a
verification query, faithful to the predicate by Theorem~\ref{thm:compile}---so a
design's two implementations, its code and its checks, are both \emph{derived} from
its denotation, each with a faithfulness guarantee.

This mirrors what Codd established for queries---but at the level of \emph{meaning}.
Codd gave queries a semantic foundation, the relational model, under which a query's
meaning is fixed independently of how it is written or which algebra evaluates it.
Grain theory does the same for pipelines: the \emph{semantic domain}
$(\grain{R},\ek{R},\BC{R})$---with computable grain-inference rules (the denotation
function) and a compositional homomorphism (\S\ref{sec:background})---settles a
pipeline's correctness by meaning, before any operation set or engine is fixed.
\pdd{} is the \emph{denotational-design discipline} over that domain; \pda{}
(\S\ref{sec:pda}) is \emph{one} algebra that realizes it---chosen because its
operations model data-engineering intent, but replaceable by relational algebra,
the dataflow model's primitives~\cite{akidau2015dataflow}, or a typed PySpark---any
operation set whose steps carry grain-inference rules (as PDA's do), so CalcG can
thread the denotation through composition (PDD Universality,
Theorem~\ref{thm:opgen}). The foundation is the semantic domain, not
the operations: the semantic domain is to \pda{} as the relational model is to
relational algebra. The payoff, developed in \S\ref{sec:correctness}, is a three-layer
correctness framework---grain, behavioral class, and domain---under a single
principle: a pipeline's correctness is a property of its \emph{code}, settled at
design time. The three layers discharge it before any data flows
(\emph{Pipeline Correctness}, Theorem~\ref{thm:pct}); what remains for runtime is not
correctness but \emph{data quality}---whether the inputs meet the design's
preconditions.

\subsection{The Pipeline Development Life Cycle}
\label{sec:overview:pdlc}

PDD structures pipeline construction as a four-step life cycle (PDLC)---a
development life cycle for data pipelines---that carries a design from semantics to
a correct, optimized implementation (Table~\ref{tab:pdlc}). The rest of the paper
follows this spine.

\begin{description}[leftmargin=0em,itemsep=2pt,style=unboxed]
\item[Step 1 (Prerequisites).] Gather, for every source and target, the data
semantics $(\grain{R},\ek{R},\BC{R})$---the output of data-source
analysis and data modeling---and document the business rules and cross-system
integration logic from requirements analysis.
\item[Step 2 (Pipeline Denotational Design).] Express the \emph{what} of the
pipeline---its transformation logic---in PDA, and the business rules as
first-order-logic predicates in the constraint language (PDL). Correctness is then
established \emph{at design time}: CalcG verifies the grain transformation, the type
checker verifies behavioral-class transitions and the well-formedness and
composition of the business-rule predicates, and the data-dependent obligations are
emitted as verification queries (\S\ref{sec:correctness}). The output is a correct
denotational design---the \emph{blueprint} the implementation must follow.
\item[Step 3 (Executable Design).] Implement the blueprint with a library whose
operations match PDA closely (1:1 in many cases), yielding an \emph{executable
design} that conforms fully to the PDD; running it---then its verification
queries---confirms the data-dependent layer on real data. Because the homomorphism
(\S\ref{sec:background}) makes the implementation's meaning equal the design's,
conformance is structural rather than incidental.
\item[Step 4 (Optimized Implementation).] Finally, tune the executable design for
operational constraints---performance, cost, scalability---selectively and only as
needed. Optimizations are applied as \emph{semantics-preserving} pipeline rewrites
(pipeline equivalence, \S\ref{sec:synthesis}), so the correctness established in
Steps~2--3 is preserved by construction.
\end{description}

\noindent Human effort concentrates on Step~1---authoring and validating the
specification (the source and target denotation plus the business rules); Steps~2--4
then \emph{derive} the design, \emph{check} it, and \emph{preserve} its correctness,
each producing evidence the engineer consumes rather than constructs
(\S\ref{sec:intro}).

\begin{table}[t]
\centering\footnotesize
\setlength{\tabcolsep}{4pt}
\caption{The Pipeline Development Life Cycle. Correctness is established once, at
design time (Step~2), and preserved through implementation (Step~3) and
optimization (Step~4).}
\label{tab:pdlc}
\begin{tabular}{@{}L{1.5cm}L{1.95cm}L{2.95cm}c@{}}
\toprule
Step & Artifact & Correctness evidence & \S \\
\midrule
1.~Prerequisites & data semantics, business rules & --- & --- \\
2.~PDD & denotational design (blueprint) & CalcG, type checker, generated queries & \ref{sec:background}--\ref{sec:correctness} \\
3.~Executable design & implementation (code) & 1:1 conformance, query results on data & \ref{sec:toolchain} \\
4.~Optimized impl. & tuned pipeline & semantics-preserving rewrites & \ref{sec:synthesis},\,\ref{sec:toolchain} \\
\bottomrule
\end{tabular}
\end{table}

\subsection{Running Example}

We thread one example throughout: a \emph{customer-order} pipeline. Order events
(grain $\textsf{CustomerId}\times\textsf{OrderDtm}$, behavioral class
\textsf{IsEvent}) are read incrementally, enriched by a lookup to a versioned
\textsf{IsMultiVersion} customer table and a product table, then aggregated into a
per-customer summary (grain $\textsf{CustomerId}$, class \textsf{IsEntity}) and
upserted into the target. This small pipeline already exercises a capture, a
grain-safe join, an aggregation that changes the grain, and a stateful apply---and,
as \S\ref{sec:correctness} shows, CalcG proves its grain correct before any data is
read, while flagging the fan-trap variant of \S\ref{sec:intro} as a design error.

% =============================================================================
% 3. BACKGROUND: GRAIN THEORY AND THE GRAIN HOMOMORPHISM   (budget ~1.5 pp)
% Sources: PDD.md Part I + §II.0/II.1; arXiv foundations.tex (grain def, fund.
% theorems, grain-determines-semantics), adt-grain.tex (universality, windowing),
% relations.tex (equality/ordering/incomparability, lattice), entity-ek-semantics.tex
% (entity, entity key, behavioral classes), dependency-theory.tex (grain lift,
% homomorphism, compositionality). Cite PODS for proofs.
% Self-contained recap of the grain-theory results the synthesis story rests on.
% =============================================================================

\section{Background: Grain Theory and the Grain Homomorphism}
\label{sec:background}

This section recalls grain theory and the homomorphism \eqref{eq:homomorphism}
that PDD rests on. This material is \emph{Part~I}~\cite{graintheory-pods}: it defines
the denotation of data and its inference rules and \emph{proves} the grain
homomorphism---the foundation this paper (Part~II) builds the methodology on. The
theory is developed in full, with proofs, there; here we give a self-contained account
of the results the synthesis story needs, with intuition in place of formal proof, and
defer every proof to~\cite{graintheory-pods}. The one result in this section that is
\emph{new}---PDD Universality (Thm.~\ref{thm:opgen})---is a contribution of Part~II and
flagged as such. The
\emph{denotation of data} it builds is a triple $(\grain{R},\ek{R},\BC{R})$---the
grain, entity key, and behavioral class of a type---that says \emph{what} each
element represents, \emph{which} entity it is about, and \emph{how} it must be read
and written.

\subsection{Grain: the Irreducible Core of a Type}
\label{sec:bg:grain}

The \emph{grain} $\grain{R}$ of a type $R$ is its level of detail: the irreducible
identifying core that uniquely determines---and so represents---any element of
$R$~\cite{graintheory-pods}. Formally, a grain is a type $G$ that is
\emph{isomorphic} to $R$ (via the \emph{grain function}
$f_g : G \stackrel{\cong}{\longrightarrow} R$) and
\emph{irreducible}: no proper sub-type of $G$ is still isomorphic to $R$. The grain
is thus the canonical minimal representative of $R$'s isomorphism class---the
\emph{atom} of the type, the smallest type carrying $R$'s full identifying
structure. Any two grains of $R$ are themselves isomorphic, so the grain is unique
up to isomorphism and every result below is independent of how it is named (e.g.\ a
\textsf{WeatherObservation} has grain $\textsf{ObservationId}$, equivalently
$\textsf{StationId}\times\textsf{ObservationDtm}$). An \textsf{OrderLineItem} has
grain $\textsf{OrderId}\times\textsf{LineItemId}$; aggregating line items by order
coarsens the grain to $\textsf{OrderId}$. Grain may be \emph{internal}---its fields
are columns of $R$, as above---or \emph{external}: a collection of monthly balances
is identified by $\textsf{BalanceDate}$, which is not a stored field. The inverse of
the grain function is the \emph{grain projection} $\grainproj{R}:R\to\grain{R}$,
which extracts each element's identifying component.

\paragraph{Type-level field-set operations.} Throughout, a (product) type is
identified with its set of fields, and grain theory lifts the set operations to
types~\cite{graintheory-pods}: $R_1\tun R_2$, $R_1\tin R_2$, and $R_1\tdiff R_2$ are
the union, intersection, and difference of field sets, and $R_1\subt R_2$ holds when
$R_1$'s fields are contained in $R_2$'s. These act on \emph{schemas, at the type
level}---no data---and are the workhorse of grain computation: the grain lattice
below, CalcG (\S\ref{sec:correctness}), the equi-join grain
$\grain{R_1}\tun(\grain{R_2}\tdiff Jk)$, and the result-type entry of every operation
contract are all written with them.

\paragraph{A guaranteed key.} The grain projection is \emph{injective}: distinct
elements have distinct grains~\cite{graintheory-pods}. The grain is therefore a
superkey that the type \emph{itself guarantees}---it plays the role of a primary
key, but as a \emph{type-level} property fixed by the schema rather than a
collection-level constraint that must be declared and that filtering or a
careless transformation can silently break. This is exactly what lets a data
collection be kept grain-unique \emph{by construction} (\S\ref{sec:pda:collection}).

\paragraph{Compositional, idempotent, universal.} The grain operator distributes
over the type constructors---$\grain{R_1\times R_2}=\grain{R_1}\times\grain{R_2}$
and $\grain{R_1+R_2}=\grain{R_1}+\grain{R_2}$---and is idempotent,
$\grain{\grain{R}}=\grain{R}$ (the grain of a grain is itself). By the same recipe,
grain extends \emph{compositionally to every algebraic data type}: products, sums,
and the inductive (lists, trees) and coinductive (streams) types built from
them~\cite{graintheory-pods}. Grain theory is in this sense \emph{type-universal}.
The one subtlety is infinite data: a stream's grain is itself infinite---finitely
\emph{presented} as a type expression, but with no finite carrier. It acquires a
finite grain the moment the stream is \emph{windowed} or \emph{partitioned}, the
grain becoming the partition key $K$ (a per-station daily window has grain
$\textsf{StationId}\times\textsf{Day}$)---which is precisely the form in which a
stream processor evaluates an unbounded input. Grain reasoning thus covers batch
and streaming uniformly: a windowed stage is verified exactly as a table is
(\S\ref{sec:pda:collection}).

\subsection{Grain Relations and the Grain Lattice}
\label{sec:bg:relations}

\begin{table*}[t]
\centering\footnotesize
\setlength{\tabcolsep}{5pt}
\caption{The five core behavioral classes~\cite{graintheory-pods}. A type's
\emph{behavior} is its grain \emph{together with the inter-element dependency the
grain incorporates}; that dependency---not the grain shape---is what separates the
classes (e.g.\ \textsf{IsEvent} vs.\ \textsf{IsMultiVersion}). The classes form a
hierarchy, $\textsf{IsMultiVersion}\subset\textsf{IsEvent}\supset\textsf{IsSeqEvent}\supset\textsf{IsSnapshot}$,
each subclass adding constraints, and they subsume the standard data-modeling
paradigms. Only \emph{transaction} facts are \textsf{IsEvent}; \emph{periodic-snapshot}
facts are \textsf{IsSnapshot}.}
\label{tab:bc}
\begin{tabular}{@{}L{1.9cm}L{2.5cm}L{6.1cm}L{5.5cm}@{}}
\toprule
Behavioral class & Grain $\grain{R}$ & Inter-element dependency incorporated in the grain & Modeling-paradigm constructs it subsumes \\
\midrule
\textsf{IsEntity} & $\ek{R}$ & none---independent entities (no ordering field in the grain) & relational (3NF) entity; Kimball SCD-1 dimension; Data Vault hub; master / reference data \\
\textsf{IsEvent} & $\ek{R}\times\textsf{EventDtm}$ & happened-before partial order; elements are immutable occurrences (append-only) & Kimball \emph{transaction} fact; event log; event-sourced stream; audit trail \\
\textsf{IsMultiVersion} & $\ek{R}\times\textsf{FromDtm}$ & effective-before order; each element is a \emph{state} valid over an interval; consecutive same-entity versions must differ in payload & Kimball SCD-2 dimension; Data Vault satellite; temporal / bitemporal table \\
\textsf{IsSeqEvent} & $\ek{R}=\textsf{EventDtm}$ & strict total order---a single global sequence, no concurrency & time series; sequential ledger; CDC stream feed \\
\textsf{IsSnapshot} & $\ek{R}=\textsf{SnapshotDtm}$ & strict total order; each element covers the complete entity population at its snapshot time & Kimball \emph{periodic-snapshot} fact; monthly balance snapshot \\
\bottomrule
\end{tabular}
\end{table*}

Three relations compare the grains of two types~\cite{graintheory-pods}. Types are
\emph{grain-equivalent}, $R_1\eqg R_2$, when their grains are isomorphic---
equivalently, when $R_1\cong R_2$: they represent the same entities at the same
detail, possibly from different perspectives. They are \emph{grain-ordered},
$R_1\leg R_2$ (``$R_1$ \emph{refines} $R_2$''), when a surjection
$\grain{R_1}\thra\grain{R_2}$ exists, so $R_1$ carries at least as much detail as
$R_2$; equivalently a one-to-many map $R_1\thra R_2$ exists. This is the structural
pattern of a foreign key (many-to-one): e.g.\ $\textsf{OrderDetail}\leg\textsf{Order}$
via $(\textsf{orderId},\textsf{lineItemId})\mapsto\textsf{orderId}$. Ordering is a
partial order up to isomorphism, and---being surjection existence at the type
level---it is the structural counterpart of a functional dependency, admitting an
Armstrong-style axiomatization that propagates declared
orderings~\cite{graintheory-pods}. Finally, $R_1$ and $R_2$ are
\emph{grain-incomparable}, $R_1\incg R_2$, when neither refines the other---the
configuration behind the \emph{fan trap} of \S\ref{sec:intro}, where two facts of
incomparable grain are joined and silently inflate an aggregate.

Under $\leg$ the space of types forms a \emph{bounded lattice}
(Figure~\ref{fig:lattice}). Any two types $R_1,R_2$ have a least upper bound
$R_1\tin R_2$, their \emph{finest common coarsening}, and a greatest lower bound
$R_1\tun R_2$, their \emph{coarsest common refinement}; on product types these are
computed by intersection and union of the grain field sets. The lattice is what
makes grain \emph{inference} compositional: each operation's output grain is
obtained from its inputs' grains by these field-set operations, and the
incomparable grains arising from a join are resolved through the lattice---the
basis of the inference rules of \S\ref{sec:bg:computable} and the \calcg{}
algorithm of \S\ref{sec:correctness}.

\begin{figure}[t]
\centering
\begin{tikzcd}[column sep=small, row sep=normal]
  & R_1 \tin R_2 \ \text{\scriptsize(coarsest)} & \\
  R_1 \arrow[ur, "\leg"] & & R_2 \arrow[ul, "\leg"'] \\
  & R_1 \tun R_2 \ \text{\scriptsize(finest)} \arrow[ul, "\leg"] \arrow[ur, "\leg"'] &
\end{tikzcd}
\caption{The grain lattice for two incomparable types. The least upper bound
$R_1\tin R_2$ (top) is the finest common coarsening; the greatest lower bound
$R_1\tun R_2$ (bottom) the coarsest common refinement. Join inference resolves
incomparable grains through these operations.}
\label{fig:lattice}
\end{figure}
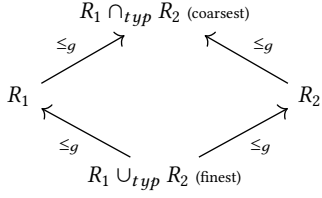

\subsection{Entity and Entity Key: a Grain Within the Grain}
\label{sec:bg:ek}

Every element is \emph{about} some subject---its \emph{entity}---and the
\emph{entity key} $\ek{R}$ is the grain of that subject
($\ek{R}\subseteq_{typ}\grain{R}$): the minimal fields identifying who or what the
element pertains to. Unlike the grain, which is computed mathematically, the entity
is a \emph{declared} semantic property, fixed by domain analysis (``what is this
data about?'')~\cite{graintheory-pods}. When the element \emph{is} the subject, the
entity key is the whole grain: an \textsf{OrderLineItem}'s subject is the line item
itself, so $\ek{R}=\grain{R}=\textsf{OrderId}\times\textsf{LineItemId}$. When the
element is an \emph{occurrence} or \emph{version} of a subject, the entity key is a
proper part of the grain---a grain within the grain: a customer version (grain
$\textsf{CustomerId}\times\textsf{FromDtm}$) is about a customer, so
$\ek{R}=\textsf{CustomerId}$. The entity key is the \emph{reconciliation point
across types of different grain}: a current-state \textsf{Customer} table (grain
$\textsf{CustomerId}$) and its SCD-2 history (grain
$\textsf{CustomerId}\times\textsf{FromDtm}$) describe the same customers at
different detail, and so must agree on entity key---\emph{entity preservation}, the
collection relation $=_{ek}$ (\S\ref{sec:correctness}), which is strictly more
informative than a cardinality check ($|c_1|=|c_2|$) because it reveals \emph{which}
entities are present or missing. This is what lets data at one grain be joined,
integrated, or augmented with data at another (\S\ref{sec:synthesis}).

\subsection{Behavioral Class: Grain Determines Behavior}
\label{sec:bg:bc}

\paragraph{One schema, many behaviors.} The same fields admit different grains, and
each grain is a different \emph{kind} of data. For \textsf{Customer} with fields
$(\textsf{CustomerId},\textsf{Name},\dots,\textsf{FromDtm},\textsf{CreatedOn})$:
\begin{itemize}[leftmargin=1.2em,itemsep=1pt,topsep=2pt]
\item grain $\textsf{CustomerId}$---an \emph{entity}: each row is a customer's
current state;
\item grain $\textsf{CustomerId}\times\textsf{FromDtm}$---\emph{versioned}: each row
is a time-stamped state, valid over an interval;
\item grain $\textsf{CustomerId}\times\textsf{CreatedOn}$---an \emph{event}: each
row is an immutable creation occurrence.
\end{itemize}
One table, three grains, three correct ways to read and write it---so without a
declared grain the data cannot be interpreted at all. In practice engineers resolve
this ambiguity by guesswork (column names, sample data, tribal knowledge), the root
cause of the transformation errors of \S\ref{sec:intro}. What the grain pins down
is the type's \emph{behavior}.

\paragraph{Behavioral class.} A type's \emph{behavioral class} $\BC{R}$ is its
\emph{behavior}: the grain \emph{together with the dependency among elements that the
grain incorporates}~\cite{graintheory-pods}. A field changes behavior only when it
enters the grain---in the example above one set of \textsf{Customer} fields becomes
independent entities, time-ordered states, or creation events purely by its grain.
This is why \emph{behavior}, not domain meaning, dictates the code. The dependency
is what separates classes that share a grain \emph{shape}: \textsf{IsEvent} elements
are immutable occurrences under a happened-before order, whereas
\textsf{IsMultiVersion} elements are \emph{states} valid over an interval under an
effective-before order, with the added invariant that consecutive versions of an
entity must differ in payload---so $\textsf{IsMultiVersion}\subset\textsf{IsEvent}$
is a genuine refinement, not a renaming. Because $\BC{R}$ is determined by grain, it
is \emph{read off the schema}; the five core classes form a hierarchy and a handful
of them \emph{unify the data-modeling paradigms}---relational/3NF, dimensional,
Data~Vault, document, time-series, event-sourced, and graph (Table~\ref{tab:bc}).
Order-indexed analogues (\textsf{IsOrdered}, \textsf{IsSeqOrdered}), an N:M relation
class (\textsf{IsNMRelation}, e.g.\ Data~Vault links and graph edges), and a
parent--child mixin (\textsf{HasParent}) extend the taxonomy. Together these form the
\emph{type-level denotation} $\sem{R}=(\grain{R},\ek{R},\BC{R})$---and they carry
precisely what the schema cannot. \emph{A schema is only structure.} The fields of
\textsf{Customer} do not, by themselves, say \emph{what} each row represents (its
level of detail, fixed by $\grain{R}$), \emph{which} entity it integrates on (its
integration point, $\ek{R}$), or \emph{how} it may be read and written ($\BC{R}$)---
the very questions any correct transformation must answer. A design that is
\emph{declarative of its intent}, and verifiable as such, therefore cannot act on
bare schemas: it must act on types that carry their denotation $\sem{R}$. That
denotation is the object the algebra of \S\ref{sec:pda} manipulates and the layers
of \S\ref{sec:correctness} verify.

\subsection{Why PDD Holds: The Grain Homomorphism}
\label{sec:background:homomorphism}

The homomorphism \eqref{eq:homomorphism} that licenses design-time reasoning is a
theorem of grain theory. Write $\grainproj{R}:R\to\grain{R}$ for the grain
projection. For any transformation $h:R_1\to R_2$, its \emph{grain lift}---the
\emph{grain component} of the transformation's denotation $\sem{h}$, capturing its
grain-to-grain content and discarding all payload---is the composite
\[
\gl{h} \;=\; \grainproj{R_2}\circ h\circ \grainproj{R_1}^{-1}
       \;:\; \grain{R_1}\to\grain{R_2}.
\]
It is what \calcg{} computes (the entity-key and behavioral-class components are
determined alongside it, \S\ref{sec:pda:classes}); we keep the distinct name to mark
that it tracks only the grain.

\begin{theorem}[Grain homomorphism~\cite{graintheory-pods}]
\label{thm:grain-hom}
For any $h:R_1\to R_2$, grain projection commutes with transformation:
$\grainproj{R_2}\circ h = \gl{h}\circ\grainproj{R_1}$
(Figure~\ref{fig:grain-hom}). Moreover the grain lift is compositional,
$\gl{h_2\circ h_1}=\gl{h_2}\circ\gl{h_1}$, and $h$ is a surjection iff $\gl{h}$ is
(so $R_1\leg R_2$ exactly when $\gl{h}$ is surjective).
\end{theorem}

\begin{figure}[t]
\centering
\begin{tikzcd}[column sep=large, row sep=normal]
  \grain{R_1} \arrow[r, "\gl{h}"] & \grain{R_2} \\
  R_1 \arrow[u, "\grainproj{R_1}"] \arrow[r, "h"'] & R_2 \arrow[u, "\grainproj{R_2}"']
\end{tikzcd}
\caption{The grain homomorphism. The two paths from $R_1$ to $\grain{R_2}$ agree:
transforming then projecting to grain equals projecting to grain then applying the
grain lift. The top row reasons about grain alone; the bottom row about full types.}
\label{fig:grain-hom}
\end{figure}
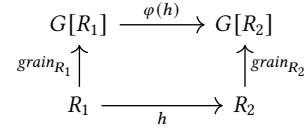

\noindent The commutativity holds by construction (the lift is \emph{built} to make
grain-level reasoning faithful); the content is the compositionality. By
Theorem~\ref{thm:grain-hom}, the grain-level meaning of a whole pipeline is the
composition of the grain-level meanings of its steps---so a designer can reason
about an entire pipeline on the top row of Figure~\ref{fig:grain-hom}, never
touching the data on the bottom row. This is precisely the homomorphism
\eqref{eq:homomorphism} of \S\ref{sec:overview}, and the foundation of
Theorem~\ref{thm:synthesis}.

\subsection{Why PDD Is Computable}
\label{sec:bg:computable}

A homomorphism alone would be only an abstract guarantee. PDD is a practical
\emph{method} because two further ingredients make the grain lift \emph{computable
without touching data}. First, grain theory supplies \emph{inference rules}: for
every relational-algebra operation there is an explicit rule that computes the
output grain from the input grains (restriction and extension preserve it,
projection and grouping change it, equi-join follows
$\grain{R_1}\cup_{typ}(\grain{R_2}-_{typ}Jk)$), each established and proved
in~\cite{graintheory-pods} (Thms.~36--42). These rules \emph{are} $\gl{h}$ made
concrete---they say exactly what each operation's denotation is, as a finite
computation on field sets, resolved through the grain lattice of
\S\ref{sec:bg:relations}. Second, because $\BC{R}$ is determined by the grain, the
denotation $(\grain{R},\ek{R},\BC{R})$ is recomputable after every step and acts as
an \emph{active constraint system}, rejecting semantically invalid operations at
design time. Thus after each operation the three questions that matter---\emph{what
is the result's grain, entity key, and behavioral class?}---are all answerable at
the type level, which is exactly what \S\ref{sec:pda}--\ref{sec:correctness}
exploit to verify pipelines at zero cost. The specific PDA operators of
\S\ref{sec:pda} are one such instantiation; any transformation language with
computable grain inference would support PDD equally. Because the inference rules
act on field sets and the pipeline DAG, never on the elements, grain reasoning is
independent of a collection's cardinality: a windowed stream is verified exactly as
a table is, finiteness being a property of \pda{}'s batch carrier rather than of the
grain theory (\S\ref{sec:pda:collection}). Part~I already observes that its \calcg{}
verification algorithm is not tied to relational algebra: it needs only a
\emph{computable} per-operation grain rule over a finite DAG, and windowed collections
supply one for streaming pipelines~\cite{graintheory-pods}. We make this a
\emph{universality} result: grain is defined for every algebraic data type
(\S\ref{sec:bg:grain}), and PDD inherits that universality over the operation set and
the carrier as well.

\begin{theorem}[PDD Universality]
\label{thm:opgen}
Let $\textsf{Op}$ be a \emph{grain-inferring} operation set---operations over
grain-theoretic types, realized over any carrier $F$ (a container type---\pda{}'s
finite-bag collection, a stream, \dots), in which every $\textsf{op}\in\textsf{Op}$
carries a \emph{faithful} grain rule $\calcg[\textsf{op}]$, i.e.\
$\sem{\textsf{op}\,c}=\calcg[\textsf{op}]\,\sem{c}$ for every input $c$. Then for every pipeline $p$ composed from $\textsf{Op}$:
\emph{(i)}~the homomorphism $\sem{p\,\textsf{src}}=\sem{p}\,\sem{\textsf{src}}$ holds;
\emph{(ii)}~$\calcg[p]$ is computable by threading the per-operation rules through the
pipeline DAG, $\calcg[\textsf{op}_1\seqc\textsf{op}_2]=\calcg[\textsf{op}_2]\circ\calcg[\textsf{op}_1]$;
and \emph{(iii)}~grain correctness---whether
$\calcg[p](\grain{\textsf{Source}})=\grain{\textsf{Target}}$---is decidable at design
time, with no data access. Design-time verification (\S\ref{sec:correctness}) therefore
applies to \emph{any} such $\textsf{Op}$, of which \pda{} is one instance. The
denotation is total on the types involved because grain extends to arbitrary algebraic
data types, inductive and coinductive alike~\cite{graintheory-pods}; and since
decidability needs only finite \emph{schemas}, not finite data, the carrier $F$ may be
unbounded.
\end{theorem}

\begin{proof}[Proof sketch]
Everything rests on the rules being \emph{exact}. Faithfulness says that running an
operation and then reading off its grain gives the same answer as applying that
operation's rule to the input's grain:
$\sem{\textsf{op}\,c}=\calcg[\textsf{op}]\,\sem{c}$. Chaining operations therefore chains
their rules: for $p=\textsf{op}_1\seqc\cdots\seqc\textsf{op}_n$, applying the rules in
order turns the source grain into the pipeline's output grain,
$\sem{p\,c}=\calcg[\textsf{op}_n]\cdots\calcg[\textsf{op}_1]\,\sem{c}$ (a one-line
induction on the chain; this is the compositionality of the grain homomorphism,
\S\ref{sec:background})---which is exactly (i) and (ii). Threading the per-operation
rules through the DAG this way \emph{is} the \calcg{} verification algorithm of Part~I:
walk the DAG in topological order, reusing a computed grain at a fan-out and applying
the binary rule at a fan-in~\cite{graintheory-pods}. For (iii), each rule is a finite
computation on field sets and the DAG is finite, so the target check
$\calcg[p](\grain{\textsf{Source}})=\grain{\textsf{Target}}$ terminates without reading
any data---in time $O(|V|\cdot k\cdot|F|)$ for field-set rules (a DAG of $|V|$
operations, fan-in $k$, at most $|F|$ fields per type)~\cite{graintheory-pods}. Finally,
a grain is defined for every algebraic data type~\cite{graintheory-pods}, and no rule
inspects the data---only the schema---so an unbounded stream is verified exactly as a
finite table. \fullproof
\end{proof}

\begin{table}[t]
\centering\small
\caption{The grain rule $\calcg[\textsf{op}](G_{\text{in}})$ is the same finite
field-set computation across operation sets and carriers: a group-by sets the grain to
its key, an equi-join unions the input grains minus the join key, a window adds itself
to the grain. \pda{} instantiates a pattern that relational algebra and the dataflow
model share.}
\label{tab:opgen}
\begin{tabular}{@{}L{2.1cm}L{2.3cm}L{3.0cm}@{}}
\toprule
Operation set (carrier) & Operation & $\calcg[\textsf{op}](G_{\text{in}})$ \\
\midrule
Relational algebra (relations) & equi-join $\bowtie_{Jk}$ & $\grain{R_1}\tun(\grain{R_2}\tdiff Jk)$ \\
 & aggregation $\gamma_{X}$ & $X$ \\
\pda{} (finite bags) & \texttt{lookup}\,$Jk$ & $\grain{R_1}\tun(\grain{R_2}\tdiff Jk)$ \\
 & \texttt{groupby-ek} & $\ek{R}$ \\
Dataflow (bounded/ unbounded) & \texttt{GroupByKey} & $K$ \\
 & windowed \texttt{GroupByKey} & $K\times\textsf{Window}$ \\
\bottomrule
\end{tabular}
\end{table}

% =============================================================================
% 4. THE PIPELINE DESIGN ALGEBRA (PDA)   (budget ~2.25 pp)
% Drafted from PDD.md §II.2 (lines 839–2693). First full draft.
% =============================================================================

\section{The Pipeline Design Algebra}
\label{sec:pda}

\pda{} is the concrete \emph{instrument} with which this paper instantiates PDD---one
grain-inferring operation set among those Theorem~\ref{thm:opgen} admits (relational
algebra and the dataflow model are others), chosen because its operations model
data-engineering intent directly. The methodology rests on the semantic domain
(\S\ref{sec:overview}), not on these particular operations; this section fixes the
vocabulary the rest of the paper is written in.

In PDD, the meaning of a concrete data operation---a SQL query, a PySpark
transformation, a dbt model---is a \emph{PDA operation}. PDA operations are
\emph{declarative}: they express the intent of a transformation (the \emph{what})
without prescribing an implementation (the \emph{how}), which is what makes a
design engine-independent and formally verifiable. Every operation acts on a
\emph{data collection}, the container we define first.

\begin{remark}[\pda{} is one instantiation]
\label{rem:pda-instantiation}
Nothing in this section is load-bearing for PDD in the way the semantic domain
(\S\ref{sec:overview}) is. Any operations that model concrete transformations and
admit computable grain inference after each step would support PDD equally
(PDD Universality, Theorem~\ref{thm:opgen})---relational algebra, the dataflow
model's primitives, or a typed PySpark all serve, over carriers ranging from relations
to unbounded streams. \pda{} is the instantiation we choose because its operations
model data-engineering intent directly. The foundation is the semantic domain---grain, behavioral class, and entity
key, with their inference rules---not these particular operations.
\end{remark}

\subsection{The Data Collection}
\label{sec:pda:collection}

Where grain theory (\S\ref{sec:background}) reasons about data \emph{types}, the
algebra acts on \emph{data collections}. Colloquially, a data collection is a
dataset with a guaranteed primary key---its grain. Formally, the meaning of a
concrete data source---a table, a Parquet file, a Spark DataFrame, a window of a
stream---is a data collection $c:\Coll{R}$: a finite, unordered bag whose elements
are kept \emph{grain-unique by construction}. It is built from two primitives,
$\emptyset_{coll}$ and an insertion that adds an element only if no element with the
same grain is already present, so every collection satisfies the grain-uniqueness
invariant
$\forall r_1,r_2\in c.\ \textit{grain}\,r_1=\textit{grain}\,r_2\Rightarrow r_1=r_2$.
The abstraction is implementation-independent: the carrier is any type constructor
$C:\mathsf{Set}\to\mathsf{Set}$---a functor---that supports the collection interface
($\emptyset_{coll}$, the grain-unique \texttt{insert-coll}, \texttt{delete-coll}, and the
\texttt{foldr} catamorphism) and maintains that invariant; lists, bags, database
tables, and DataFrames are all instances. Type extractors---\texttt{GetG},
\texttt{GetEK}, \texttt{GetBC}---then recover its semantics $(\grain{R},\ek{R},\BC{R})$
from the schema alone, so designs are portable across engines. The collection is not part of the grain-theory
denotation, which is collection-free (\S\ref{sec:background})---it is the
\emph{carrier} that \pda{}'s operation set commits to. Grain theory assumes only a
bare ``finite bag'' to state its inference rules over relational algebra operations,
and a grain denotation never inspects the carrier; \pda{}, by contrast, \emph{needs} one: the write operations of
\S\ref{sec:pda:apply} are element-level merges, and the class invariants below
($\subseteq_c$, completeness, grain-unique maintenance) are collection-level---none
is visible to grain alone. The carrier is \emph{finite} not because grain theory
demands it, but because \pda{}'s operations are \emph{batch} captures, transforms,
and applies over datasets; a streaming operation set would denote the same grains
over an unbounded carrier. Only at execution---the Executable Design
(\S\ref{sec:toolchain})---is the carrier realized concretely, as a typed finite
collection such as a DataFrame.

\subsection{Four Morphism Classes}
\label{sec:pda:classes}

Every PDA operation belongs to exactly one of four classes. Three act on
\emph{data}---reading, transforming, and writing---and a fourth acts only at the
\emph{type} level:
\[
\begin{array}{ll}
\textsf{IsCapture}\;A = A \to A & \text{read (endomorphism)}\\
\textsf{IsTransform}\;A\,B = A \to B & \text{transform (pure function)}\\
\textsf{IsDeclare}\;A\,B = A \to B & \text{re-declare (type-level only)}\\
\textsf{IsApply}\;A = A \to A \to A & \text{write (binary merge)}
\end{array}
\]
where $A,B$ range over arbitrary types; every catalog operation instantiates them
at a collection ($A=\Coll{R}$). Leaving the classes polymorphic, rather than fixing
the carrier to $\Coll{\cdot}$, is what lets a pipeline range over a composite type
such as the product built by $\parc$ (\S\ref{sec:synthesis}).
A capture filters but never adds or alters elements ($\textsf{op}\,c \subseteq_c c$).
A transform is a pure, total function that may change the type, grain, and
behavioral class. A \emph{declare} changes only the denotation---the declared grain,
behavioral class, and derived entity key---performing \emph{no} operation on the
data; it is the one family with no relational-algebra analogue, sound under a
precondition (its declared grain is a key of the input) discharged at design time
(\S\ref{sec:correctness:regrain}). An apply merges incoming data into a
before-state, with an explicit conflict-resolution rule.

The three \emph{data} acts are the stages of the \emph{Capture--Transform--Apply}
(CTA) pipeline anatomy that \S\ref{sec:synthesis} composes into pipelines; a declare
acts only on the type and adds no stage, as its composition shows.

\begin{proposition}[Declares compose away]
\label{prop:declare-composes}
A declare changes the denotation but touches no data, so on its own it is neither a
capture (which keeps the denotation) nor a transform (which touches data). Its
precondition---the declared grain is a key of the input---is established only by a
\emph{data} operation that makes that grain a key (a capture or transform such as
\texttt{pick}, \texttt{checkout}, \texttt{groupby}); another declare, doing no data
work, cannot supply it, so declares do not chain. Hence a grain-changing declare is
always adjacent to a capture or transform, and the composite ($C\mathbin{\circ}D$,
$D\mathbin{\circ}C$, $T\mathbin{\circ}D$, $D\mathbin{\circ}T$) both touches data and
changes the type---a transform. A declare therefore introduces no pipeline stage: it
folds into a process-phase transform, leaving the CTA anatomy intact.
\end{proposition}

\paragraph{Grain, lens, and key.} The three type-level properties are \emph{total}:
every collection has a grain $\grain{R}$, a behavioral class $\BC{R}$, and an entity
key $\ek{R}$---none is ever absent (a scalar aggregate such as a customer count is
its own irreducible grain, $\ek{R}=\grain{R}$, an $\textsf{IsEntity}$; a
sequential-event log has $\ek{R}=\grain{R}$ equal to the event timestamp, which
indexes the successive states of its single global entity). The grain does not pin
the BC down by itself; it offers a
\emph{finite menu} of admissible lenses, each gated by a dependent-type
constraint---the grain $\textsf{CustomerId}\times\textsf{Updated\_at}$ admits
$\textsf{IsEntity}$, or $\textsf{IsEvent}$/$\textsf{IsMultiVersion}$ on
$\textsf{Updated\_at}$ when their constraints hold---and the engineer \emph{declares}
which lens to view the data through (a \texttt{declare-$\langle$bc$\rangle$}---a
\emph{declare}, \S\ref{sec:pda:declare}---can re-lens a type even when nothing else
about it changes). The entity key is then \emph{derived} from
the grain and the chosen BC by that class's own rule, so the determination order is
grain $\to$ BC $\to$ EK; the contract records BC as \emph{declared}, EK as
\emph{derived}. Hence an operation that preserves the grain and keeps the lens
preserves BC and EK automatically---captures and applies are BC-correct for
free---and only a grain- or lens-changing transform must declare a new BC, whose
admissibility is discharged at the type level by the dependent-type constraint
gating that lens (\S\ref{sec:correctness}). The whole denotation
$(\grain{R},\ek{R},\BC{R})$ is thus fixed and checked at design time; only
\emph{domain} constraints (Layer~3) are ever data-dependent, and those are business
rules over values, not the denotation.

\subsection{Operation Contracts}
\label{sec:pda:contracts}

Every operation carries a \emph{contract}: a fixed-schema table stating its
precondition, postcondition, and the way it transforms type, grain, entity key,
behavioral class, element/entity containment, cardinality, keys, and which
algebraic laws it satisfies (Table~\ref{tab:contract-schema}). Contracts are not
documentation---they are the proof obligations that the three correctness layers
of \S\ref{sec:correctness} discharge. To keep the catalog compact, each class
states its \emph{defaults} once; an operation's contract lists only deviations.

\begin{table}[t]
\centering\small
\caption{The operation-contract schema. Every PDA operation instantiates these
rows; class defaults are stated once per class and only deviations are repeated.}
\label{tab:contract-schema}
\begin{tabular}{@{}ll@{\hspace{8pt}}L{4.3cm}@{}}
\toprule
Group & Row & Meaning \\
\midrule
Conditions & Pre / Post & input requirements; guarantees on the result \\
Type-level & Result type & how $F$ relates to $R$ ($=R$, $\subt R$, $R\tun D$, \dots{}) \\
           & Grain & $\grain{\text{Res}}$: preserved / declared / computed \\
           & BC & $\BC{\text{Res}}$: preserved / declared (the lens) \\
           & Entity key & $\ek{\text{Res}}$: preserved / derived (from grain $+$ BC) \\
Coll-level & Elements & $\text{Res}\subseteq_c c$, or new elements constructed \\
           & Entities & entity-key containment \\
           & Cardinality & bounds on $\lVert\text{Res}\rVert$ \\
           & Keys & key/superkey preservation \\
Algebraic  & Laws & idempotent, assoc., commutative, identity, \dots \\
\bottomrule
\end{tabular}
\end{table}

\subsection{Capture: Reading Data}
\label{sec:pda:capture}

Captures are endomorphisms; each behavioral class has \emph{canonical} captures
that enforce its intended read pattern (Table~\ref{tab:capture}). Reading an
\textsf{IsEvent} log is a time-window \texttt{delta}; reading the current state of
an \textsf{IsMultiVersion} table is a \texttt{checkout} at a timestamp---the
point-in-time query familiar from Kimball SCD-2 dimensions and
Data~Vault satellites; a
full scan of a versioned table is deliberately \emph{not} canonical---it must be
requested explicitly through the converter pattern below. All captures share the
defaults: result type $=R$, grain/EK/BC preserved, $\text{Res}\subseteq_c c$, and
$\lVert\text{Res}\rVert\le\lVert c\rVert$.

\begin{table}[t]
\centering\small
\caption{The capture catalog. Each capture is canonical for the behavioral
classes shown; all preserve grain and BC.}
\label{tab:capture}
\begin{tabular}{@{}lL{2.2cm}p{3.0cm}@{}}
\toprule
Operation & Canonical for & Reads \\
\midrule
\texttt{full} & \textsf{IsEntity}, \textsf{IsNMRelation} & the entire collection \\
\texttt{cdc} & \textsf{IsEntity} & rows new/changed vs.\ a before-state \\
\texttt{delta(lo,hi)} & \textsf{IsEvent} & events in a time window \\
\texttt{snapshot(t)} & \textsf{IsSnapshot} & the snapshot at time $t$ \\
\texttt{checkout(t)} & \textsf{IsMultiVersion} & latest version per entity $\le t$ \\
\texttt{asof(t)} & poly. & BC-dispatched ``state as of $t$'' \\
\texttt{filter(P)} & any & rows satisfying $P$ \\
\texttt{pick($G$,\textit{ord},\textit{dir})} & any & top-ranked row per $G$-partition \\
\bottomrule
\end{tabular}
\end{table}

\paragraph{Converter pattern.} To apply a non-canonical capture, first convert the
behavioral class with a declare (\S\ref{sec:pda:declare}). For example,
$\texttt{full}\circ\texttt{declare-entity}$ reads \emph{all} entities of a
multi-version collection---making an otherwise dangerous full read of a temporal
table explicit and auditable, rather than an accident.

\paragraph{\texttt{pick} generalizes \texttt{checkout}.}
$\texttt{pick}(G,\textit{ord},\textit{dir})$ keeps the top-ranked row per
$G$-partition---an endomorphism (the result rows are unchanged, just fewer) that
\emph{preserves} grain and BC while establishing $G$ as a key of the result.
$\texttt{checkout}$ is its \textsf{IsMultiVersion} instance (partition by entity,
rank by version time); $\texttt{pick}$ lifts the pattern to any type and ordering.
This is why it is a capture, not a transform: it selects rows rather than declaring
a new grain. Restoring a coarser grain after a \texttt{-low} fan-out is then a
$\texttt{pick}$ followed by a re-lens (\S\ref{sec:correctness:regrain})---the
selection establishes the key, the re-lens declares the grain.

\subsection{Apply: Writing Data}
\label{sec:pda:apply}

Apply operations merge incoming data \texttt{new} into a before-state
\texttt{before}. All satisfy two invariants---\emph{completeness} (every source
grain appears in the result) and \emph{idempotency} (re-applying the same source
is a no-op)---and are distinguished by a \emph{conflict-resolution rule} for
grains present in both arguments (Table~\ref{tab:apply}). The catalog reveals a
clean symmetry: \texttt{append} (first-writer-wins) and \texttt{upsert}
(last-writer-wins) differ by exactly one \texttt{delete-coll}:
\[
\begin{aligned}
\texttt{append} &= \textsf{foldr}\;\texttt{insert-coll},\\
\texttt{upsert} &= \textsf{foldr}\;(\texttt{insert-coll}\circ\texttt{delete-coll}).
\end{aligned}
\]
The multi-version applies (\texttt{commit-append}, \texttt{commit-upsert})
implement the SCD~Type~2 pattern: append a new version only when the payload
changed, and (for \texttt{commit-upsert}) close the previous version.

\paragraph{Apply is pure: no mutable state, no state logic.} Although an apply
\emph{writes}, it is a pure function $\texttt{apply}:A\to A\to A$ of an
\emph{immutable} before-state and source, returning a \emph{fresh} collection
$r=\texttt{apply}\;\texttt{before}\;\texttt{new}$; nothing is updated in place. Its
stateful condition $\textsf{Intra}$ (Thm.~\ref{thm:pct}) is therefore a plain
predicate relating the two values in hand---for a no-salary-cut policy on an
\textsf{IsEntity} target,
\[
\begin{aligned}
\textsf{Intra}(\texttt{before}, r)\;\equiv\;
&\forall e\in r,\,e_0\in\texttt{before}.\;\textit{key}(e)=\textit{key}(e_0)\\
&\Rightarrow\; e.\textsf{salary}\ge e_0.\textsf{salary},
\end{aligned}
\]
discharged by comparing the result $r$ against \texttt{before} directly. Because the
before-state is an explicit \emph{value} rather than a mutated cell, the aliasing and
interference that motivate state logics (separation logic, frame rules) cannot
arise: referential transparency gives non-interference for free, and predicates over
distinct collections compose with $\wedge$, not a separating conjunction. This is why
the same Hoare-style, refinement/dependent-typed discipline (\S\ref{sec:related})
covers \emph{every} morphism class, writes included, with no state-logic caveat.

\begin{table}[t]
\centering\small
\caption{The apply catalog. Conflict resolution distinguishes the operations;
all preserve grain and BC and are idempotent.}
\label{tab:apply}
\begin{tabular}{@{}lll@{}}
\toprule
Operation & Canonical for & Conflict resolution \\
\midrule
\texttt{append} & \textsf{IsEvent} & first-writer-wins \\
\texttt{upsert} & \textsf{IsEntity} & last-writer-wins \\
\texttt{replace} & \textsf{IsEntity} & stateless (discard before) \\
\texttt{merge-delete} & \textsf{IsEntity} & symmetric difference \\
\texttt{commit-append} & \textsf{IsMultiVersion} & append if payload changed \\
\texttt{commit-upsert} & \textsf{IsMultiVersion} & SCD2: append + close prior \\
\bottomrule
\end{tabular}
\end{table}

\subsection{Transform: Reshaping Data}
\label{sec:pda:transform}

Transforms are the algebra's expressive core. They subdivide into families
(Table~\ref{tab:transform}), of which the joins are the most consequential for
correctness. Each transform is a pure, \texttt{foldr}-derived function with a
denotational spec; since the value-level meaning is classical relational algebra,
we summarize the catalog by what composition and correctness turn on---its effect
on grain, behavioral class, and type.

\paragraph{Set operations.} \texttt{union}, \texttt{diff}, \texttt{intersect},
\texttt{sym-diff} (and \texttt{id}) are grain-based, same-type, same-BC, and
grain-preserving; under an agreement precondition they form a Boolean algebra.

\paragraph{Joins.} Every join assigns the two collections distinct roles---a
\emph{primary} $R_1$ (the one flowing through the pipeline) and a \emph{supplementary}
$R_2$---and its result grain is fixed at the type level by the join grain-inference
theorems of~\cite{graintheory-pods} (principally the equi-join grain theorem,
Thm.~37). The grain-ordered equi-joins vary along two axes. \emph{Row preservation}
separates \texttt{lookup}, which keeps a primary row only when it matches (inner),
from \texttt{augment}, which keeps every primary row and leaves the supplementary
fields empty when it does not (left-outer). \emph{Grain direction} is set by the
ordering: the default forms require $R_1 \leg R_2$---the primary is the finer
side---whose witness \emph{is} the join function, so no key is named and each primary
row has \emph{at most one} match; the result then preserves the primary grain,
$\grain{\textit{Res}}=\grain{R_1}$---no fan-out, no fan trap, by construction. The
\texttt{-low} variants invert the ordering ($R_2 \leg R_1$): the primary descends to
the finer supplementary grain, $\grain{\textit{Res}}=\grain{R_2}$, a deliberate
fan-out the engineer declares by choosing the variant. The free equi-joins
\texttt{lookup-on}/\texttt{augment-on}---the explicit-key counterparts of
\texttt{lookup}/\texttt{augment}, with the \texttt{on}\,$jk_1\,jk_2$ clause as
sugar---and \texttt{natural-join} (keyed on shared columns) take the join key
directly and need no grain ordering.
\texttt{semi-join} and \texttt{anti-join} keep the primary rows that respectively do
and do not match; each returns a sub-collection of the primary, so
$\grain{\textit{Res}}=\grain{R_1}$. A \texttt{theta-join}, finally, matches on an
arbitrary predicate---no key relates the inputs---so the result takes the Cartesian
grain $\grain{R_1}\times\grain{R_2}$.

\paragraph{Selections.} The value selections are grain-preserving projections on the
data: \texttt{sel-safe} drops payload while keeping all grain columns (grain and BC
preserved), and \texttt{sel-grain}/\texttt{sel-ek} project to the bare
grain/entity-key \emph{values}. All are genuine projections and always well-defined.
The grain- and BC-\emph{changing} re-lenses---\texttt{declare-entity},
\texttt{declare-event}, \dots---are a separate class: they touch no data and are
treated in \S\ref{sec:pda:declare}.

\paragraph{Enrichment, grouping, and structure.} \texttt{enrich} extends each
element with computed or window-derived columns (grain-preserving);
\texttt{groupby-*} aggregates, declaring the grouping columns as the new grain and
BC; and \texttt{explode} unnests a nested field, inheriting the nested type's grain
under a grain-disjointness precondition. (Row selection per partition---\texttt{pick}---is
a capture, \S\ref{sec:pda:capture}, not a transform.) Together with the captures,
declares (\S\ref{sec:pda:declare}), and applies, these operations form the complete
vocabulary that pipeline designs are written in; every one of them comes with the
contract that the next sections exploit.

\begin{table*}[t]
\centering\footnotesize
\caption{The transform catalog: result grain $\grain{\text{Res}}$ and behavioral
class $\BC{\text{Res}}$ for each operation (``pres.''~$=$ preserved from the input).
The \texttt{groupby-$\langle$bc$\rangle$} family (like the declares of
\S\ref{sec:pda:declare}) has one operation per behavioral class; the
declared class fixes the result grain---\textsf{IsEntity}: $\ek{R}$;
\textsf{IsEvent}: $\ek{R}\times\textit{EventDtm}$; \textsf{IsMultiVersion}:
$\ek{R}\times\textit{FromDtm}$; \textsf{IsSeqEvent}: $\textit{EventDtm}$;
\textsf{IsSnapshot}: $\textit{SnapshotDtm}$; \textsf{IsOrdered}:
$\ek{R}\times\textit{OrdType}$; \textsf{IsSeqOrdered}: $\textit{OrdType}$;
\textsf{IsNMRelation}: $G_1\times G_2$.}
\label{tab:transform}
\begin{tabular}{@{}p{3.3cm}p{3.3cm}p{1.7cm}p{7.7cm}@{}}
\toprule
Operation & $\grain{\text{Res}}$ & $\BC{\text{Res}}$ & Description \\
\midrule
\multicolumn{4}{@{}l}{\emph{Set operations} (same type)}\\
\texttt{id} & $\grain{R}$ & pres. & identity \\
\texttt{union}, \texttt{diff}, \texttt{intersect}, \texttt{sym-diff} & $\grain{R}$ & pres. & grain-based Boolean algebra \\
\addlinespace
\multicolumn{4}{@{}l}{\emph{Joins} (primary $R_1$, supplementary $R_2$)}\\
\texttt{lookup} & $\grain{R_1}$ & $\BC{R_1}$ & grain-safe inner join; $R_1\leg R_2$ forces $J_k=\grain{R_2}$ (at most one match); unmatched primary rows dropped \\
\texttt{augment} & $\grain{R_1}$ & $\BC{R_1}$ & grain-safe left-outer join; every primary row kept \\
\texttt{lookup-low}, \texttt{augment-low} & $\grain{R_2}$ & $\BC{R_2}$ & inverted ordering ($R_2\leg R_1$); primary descends to the finer grain \\
\texttt{semi-join}, \texttt{anti-join} (incl.\ \texttt{-on}) & $\grain{R_1}$ & $\BC{R_1}$ & keep primary rows with / without a match ($\text{Res}\subseteq_c c_1$) \\
\texttt{lookup-on}, \texttt{augment-on}, \texttt{natural-join} & $\grain{R_1}\tun(\grain{R_2}\tdiff J_k)$ & follows grain & explicit-key equi-join (Thm.~37), no ordering required; equals $\grain{R_1}$ iff $J_k$ covers $\grain{R_2}$. \texttt{natural-join}: $J_k=R_1\tin R_2$ \\
\texttt{theta-join} & $\grain{R_1}\times\grain{R_2}$ & follows grain & arbitrary predicate $\theta$; Cartesian grain \\
\addlinespace
\multicolumn{4}{@{}l}{\emph{Selections}}\\
\texttt{sel-safe} & $\grain{R}$ & pres. & drop payload, keep all grain columns \\
\texttt{sel-grain} & $\grain{R}$ & pres. & project to grain values \\
\texttt{sel-ek} & $\ek{R}$ & \textsf{IsEntity} & project to entity-key values \\
\addlinespace
\multicolumn{4}{@{}l}{\emph{Enrichment} (grain-preserving)}\\
\texttt{enrich}, \texttt{enrich-metadata}, \texttt{enrich-window} & $\grain{R}$ & pres. & extend each row with computed, metadata, or window-derived fields \\
\addlinespace
\multicolumn{4}{@{}l}{\emph{Grouping and aggregation}}\\
\texttt{agg} & $\grain{A}$ & \textsf{IsEntity} & scalar aggregate: whole collection $\to$ one row \\
\texttt{groupby-$\langle$bc$\rangle$} & declared & $\langle$bc$\rangle$ & partition by the grouping key (the new grain) and aggregate \\
\addlinespace
\multicolumn{4}{@{}l}{\emph{Structural}}\\
\texttt{explode} & $\grain{S}$ & $\BC{S}$ & unnest a nested-collection field; grows cardinality (sub-collections grain-disjoint) \\
\bottomrule
\end{tabular}
\end{table*}

\subsection{Declare: Re-declaring the Denotation}
\label{sec:pda:declare}

A \texttt{declare-$\langle$bc$\rangle$} changes only a collection's
\emph{denotation}---its grain, behavioral class, and derived entity key---and
performs \emph{no} operation on the data (Table~\ref{tab:declare}). It is the one
operation family with no relational-algebra counterpart: where a transform reshapes
tuples, a declare re-interprets the same rows under a new grain. Each takes the
projections naming the new grain components and is sound under a single
\emph{precondition}---the declared grain is a key of the input---discharged at design
time by the preceding operation's \emph{Keys} postcondition and checked by the type
checker; nothing is read (\S\ref{sec:correctness:regrain}). This is the operational
core of denotational design: a grain is \emph{declared and verified}, not computed.

\begin{table}[t]
\centering\small
\caption{The declare catalog: one \emph{re-lens} per behavioral class. Each declares
the grain and BC shown (and derives the entity key) with no data operation, under the
precondition that the declared grain is a key of the input
(\S\ref{sec:correctness:regrain}).}
\label{tab:declare}
\begin{tabular}{@{}lll@{}}
\toprule
Operation & Declares grain & Declares BC \\
\midrule
\texttt{declare-entity}       & $\ek{R}$                        & \textsf{IsEntity} \\
\texttt{declare-event}        & $\ek{R}\times\textit{EventDtm}$ & \textsf{IsEvent} \\
\texttt{declare-multiversion} & $\ek{R}\times\textit{FromDtm}$  & \textsf{IsMultiVersion} \\
\texttt{declare-seq-event}    & $\textit{EventDtm}$             & \textsf{IsSeqEvent} \\
\texttt{declare-snapshot}     & $\textit{SnapshotDtm}$          & \textsf{IsSnapshot} \\
\texttt{declare-ordered}      & $\ek{R}\times\textit{OrdType}$  & \textsf{IsOrdered} \\
\texttt{declare-seq-ordered}  & $\textit{OrdType}$              & \textsf{IsSeqOrdered} \\
\texttt{declare-nm-relation}  & $G_1\times G_2$                 & \textsf{IsNMRelation} \\
\bottomrule
\end{tabular}
\end{table}

\noindent The precondition is the declare's only contract with the rest of the
pipeline, and it is deliberately \emph{indifferent to how the key was
established}---this is the strength of the approach. Whether the preceding step is a
\texttt{checkout}, a \texttt{pick}, a \texttt{groupby}, or an arbitrarily long
sub-pipeline, all the declare requires is that step's \emph{Keys} postcondition---that
the target grain is a key of its output; the declare neither knows nor cares what
produced that guarantee. Only when no such guarantee holds must the collapse first be
\emph{defined}, by any operation whose contract establishes the key
(e.g.\ \texttt{checkout}, \texttt{pick}, \texttt{groupby}, or bespoke), before the
re-lens applies.

\paragraph{Worked example: a versioned dimension to its current state.} The four
classes compose in one line---using pipe-forward $\mathbin{|>}$ and forward
composition $\seqc$ (Table~\ref{tab:pipeline-operators}, \S\ref{sec:synthesis}). Let
$\textsf{src}:\Coll{\textsf{CustVer}}$ be an SCD2 customer dimension---an
$\textsf{IsMultiVersion}$ type with denotation
$(\textsf{CustID}\times\textsf{FromDtm},\ \textsf{CustID},\ \textsf{IsMultiVersion})$.
To read each customer's current state as of time $t$:
\[
\begin{aligned}
&\textsf{src}\ \mathbin{|>}\ \texttt{checkout}\,t\ \seqc\
\texttt{declare-entity}\,\textsf{cust\_id}\\
&\qquad{}\seqc\ \texttt{sel-safe}\,\textsf{payload}.
\end{aligned}
\]
$\texttt{checkout}\,t$ (a capture, Table~\ref{tab:capture}) keeps the latest version
per entity $\le t$, preserving the grain but \emph{establishing $\textsf{CustID}$ as a
key} of its result. $\texttt{declare-entity}\,\textsf{cust\_id}$ (a declare,
Table~\ref{tab:declare}) re-lenses to
$(\textsf{CustID},\textsf{CustID},\textsf{IsEntity})$---its precondition, that
$\textsf{CustID}$ is a key, discharged by checkout's postcondition at the type level,
with no data touched. Finally $\texttt{sel-safe}\,\textsf{payload}$ (a transform,
Table~\ref{tab:transform}) drops the now-redundant $\textsf{FromDtm}$---a safe payload
reduction, since the declared grain $\textsf{CustID}$ is retained. The result type is
$\Coll{\textsf{CustVer}\tdiff\textsf{FromDtm}}$ with denotation
$(\textsf{CustID},\textsf{CustID},\textsf{IsEntity})$: one current-state row per
customer, an entity. Every step's denotation---and the capture$\to$declare
hand-off---is settled at design time, before any data is read.

% =============================================================================
% 5. SYNTHESIS: COMPOSING CORRECT PIPELINES   (budget ~1.75 pp)  [CORE]
% Drafted from PDD.md §II.3 (lines 2694–3645). First full draft.
% =============================================================================

\section{Synthesizing Pipelines by Composition}
\label{sec:synthesis}

\S\ref{sec:pda} gave the operations; this section composes them into pipelines.
This is the \emph{synthesis} side of the verification--synthesis duality: where
CalcG (\S\ref{sec:correctness}) \emph{verifies} that a given DAG is grain-correct,
the operators here \emph{construct} a pipeline as a single well-typed function
$\textsf{Source}\to\textsf{Target}$ whose correctness is guaranteed by the way it
was built.

\subsection{The Pipeline Type and the CTA Constructor}

A pipeline is a function from a before-state and a source to a new target state:
\[
\textsf{Pipeline}\;\textsf{Source}\;\textsf{Target}
  = \textsf{Target} \to \textsf{Source} \to \textsf{Target},
\]
where $\textsf{Source}$ and $\textsf{Target}$ are arbitrary types. The CTA
constructor instantiates them at collections ($\Coll{R}$); keeping the type
polymorphic lets the pipeline product $\parc$ range over a product
$\textsf{S}_1\times\textsf{S}_2$---a pair of collections, not a collection of pairs.
The before-state enables stateful applies (\texttt{upsert}, \texttt{append}); a
stateless pipeline simply passes $\emptyset_{coll}$, written $[\,\textsf{pip}\,]_{SL}$.
The canonical way to build a pipeline is the \emph{CTA constructor}, which wires a
capture, a transform, and an apply into the data flow
\[
\begin{aligned}
&\cta{\text{capture}}{\text{transform}}{\text{apply}} = \lambda\,\textsf{trg}\,\textsf{src}.\\
&\quad \textsf{src} \mathbin{|>} \text{capture} \seqc \text{transform} \seqc \text{apply}\;\textsf{trg},
\end{aligned}
\]
where $|>$ is pipe-forward ($x \mathbin{|>} f = f\,x$) and $\seqc$ is forward
composition ($(f\seqc g)\,x = g(f\,x)$). Read left to right, the source is
captured, transformed, then applied to the target---exactly the engineer's mental
model of a data flow.

\subsection{Composition Operators}

A handful of operators express any pipeline topology
(Table~\ref{tab:pipeline-operators}). Forward composition $\seqc$ chains the
operations inside a CTA and, once before-states are supplied, chains whole pipelines
(a stateless medallion Landing\,$\to$\,Integration\,$\to$\,Serving is
$[\,p_1\,]_{SL} \seqc [\,p_2\,]_{SL} \seqc p_3$). Two operators combine independent
pipelines: the pipeline product $\parc$ splits a product source across two branches
(parallel fan-out), while the pipeline fork $\forkc$ feeds \emph{one} source to both
branches and pairs their targets---the way a single pattern maintains two collections
at once, e.g.\ a current-state table and its history. Pipeline sequencing $\seqp$
chains sub-pipelines by materializing each prefix statelessly while the last stage
stays stateful; it is the idiomatic way a composite pattern wires its stages. Finally,
$\textsf{flattenp}$ restructures the nested tuples that products create. The three
topology-forming operators already suffice: \emph{any} DAG---sequential, parallel, or
arbitrary fan-in/fan-out---is a combination of $\seqc$, $\parc$, and
$\textsf{flattenp}$ ($\forkc$ and $\seqp$ are definable from them), reducing a whole
pipeline to one composable function. The operators obey the expected laws---$\seqc$
is associative with identity \texttt{id}, and $\parc$ distributes over $\seqc$---so
designs can be refactored and optimized with machine-checkable equivalences.

\begin{table}[t]
\centering\small
\caption{Pipeline operators. Any DAG topology is expressible with $\seqc$,
$\parc$, and $\textsf{flattenp}$; the fork $\forkc$ and pipeline-sequencing
$\seqp$ are derived conveniences for dual-target and multi-stage patterns.}
\label{tab:pipeline-operators}
\begin{tabular}{@{}lp{5.9cm}@{}}
\toprule
Operator & Meaning \\
\midrule
$x \mathbin{|>} f$ & pipe-forward: $f\,x$ \\
$f \seqc g$ & forward composition $g\circ f$: chains operations, and
before-state-applied pipelines \\
$[\,\textsf{pip}\,]_{SL}$ & stateless projection: $\textsf{pip}\;\emptyset_{coll}$ \\
$\textsf{pip}_1 \parc \textsf{pip}_2$ & pipeline product: split a product source
across two branches \\
$\textsf{pip}_1 \forkc \textsf{pip}_2$ & pipeline fork: feed one source to both
branches, pair the two targets \\
$\textsf{pip}_1 \seqp \textsf{pip}_2$ & pipeline sequencing: prefix stateless,
last stage stateful \\
$\textsf{flattenp}$ & flatten nested product tuples \\
$f \mathbin{>\!>\!>} g$ & proof-carrying composition (\S\ref{sec:correctness}) \\
\bottomrule
\end{tabular}
\end{table}

\subsection{Well-Typed Composition is Grain-Correct}
\label{sec:synthesis:thm}

The payoff of building pipelines this way is that grain correctness is not checked
after the fact---it is a property of well-typedness. Each operation's contract
fixes how it transforms grain (\S\ref{sec:pda:contracts}), and CalcG composes
these rules functionally (\S\ref{sec:correctness}). Hence:

\begin{theorem}[Compositional grain correctness]
\label{thm:synthesis}
For pipelines $\textsf{pip}_1,\textsf{pip}_2$ and the operators of
Table~\ref{tab:pipeline-operators},
$\calcg[\textsf{pip}_1 \seqc \textsf{pip}_2] = \calcg[\textsf{pip}_2]\circ\calcg[\textsf{pip}_1]$
and
$\calcg[\textsf{pip}_1 \parc \textsf{pip}_2] = \calcg[\textsf{pip}_1]\times\calcg[\textsf{pip}_2]$.
Consequently, if every stage's contract is satisfied, the composed pipeline's
output grain is determined---and a design whose computed grain equals the declared
target grain is grain-correct by construction (Thm.~\ref{thm:calcg}).
\end{theorem}

\noindent Because composition of grain rules mirrors composition of operations, a
designer never reasons about the whole pipeline at once: local well-typedness at
each stage yields global grain correctness. This is the precise sense in which PDD
is to data pipelines what type-safe query construction is to SQL.

\subsection{Pipeline Patterns}
\label{sec:synthesis:patterns}

Most production pipelines are instances of a few recurring shapes. Since a pipeline
\emph{denotes a function} (\S\ref{sec:overview},\,\S\ref{sec:synthesis}), a
\emph{pattern} denotes a \emph{parametrically polymorphic} function---a map from data
semantics ($\grain{\cdot}$, $\ek{\cdot}$, $\BC{\cdot}$ of its sources and target)
to a concrete CTA pipeline---whose operations are fixed but whose types, grains,
and arities vary. Patterns consume their semantic constraints as \emph{evidence}:
\texttt{append} consumes \textsf{IsEvent} evidence, \texttt{lookup} consumes
grain-ordering evidence. Variable arity (e.g.\ $n$ providers, an FK chain of depth
$k$) is handled by the \texttt{iter} combinator,
$\texttt{iter}\;\textsf{op}\;[c_1,\dots,c_n] = \textsf{op}\,c_1 \seqc \cdots \seqc \textsf{op}\,c_n$,
over which CalcG distributes. A pattern is \emph{correct} when every valid
instantiation is correct---verified once, symbolically, on the parametric types,
then reused for free at every instantiation. PDA ships three named patterns:

\begin{itemize}[leftmargin=1.2em,itemsep=2pt]
\item \textbf{Ingestion} ($\text{external source}\to\text{staging}$): a full or
incremental extract. The extract mode fixes the apply---\texttt{replace} (full
refresh) or \texttt{append} (incremental)---and, critically, may change the
staging collection's behavioral class.
\item \textbf{Staging-to-Landing} ($\text{staging}\to\text{persistent landing}$): promotes
the staged data to a \emph{dual target}---a current-state collection and a history
collection---written together with the fork $\forkc$. The staging BC fixes each
target's BC and apply: an entity source, for instance, yields an \textsf{IsEntity}
current state maintained by \texttt{upsert} and an \textsf{IsEvent} history maintained
by \texttt{append}, with the cross-target invariant that the latest point-in-time of
the history equals the current state.
\item \textbf{Augmentation} ($\text{Landing}\to\text{Integration}\to\text{Serving}$): the
medallion workhorse, and itself the composition---via $\seqp$---of \emph{four}
pre-verified sub-patterns,
{\small\[
\begin{aligned}
\textsf{augmentation} = {}&\textsf{delta-unif.}\seqp\textsf{delta-augment}\\
  &\seqp\textsf{stage-prep}\seqp\textsf{target-load},
\end{aligned}
\]}%
each of which is again a CTA. \emph{Delta-unif.}\ builds a target-grain \emph{driver}
delta---the target records to (re)process this run---by navigating every provider to
the target entity key according to its kind: \emph{direct} providers carry the full
target EK (projected with \texttt{sel-ek}), \emph{semi-direct} providers carry the
target entity but lack some EK columns (completed row-preservingly via
\texttt{augment}/\texttt{enrich}), and \emph{indirect} providers do not carry the
target entity (navigated via \texttt{lookup}, dropping rows that fail), then unifies
them with \texttt{iter union}. \emph{Delta-augment} adds the business content, joining
in the augmentation collections with \texttt{iter augment}. \emph{Stage-prep} enriches
and re-lenses the result to the target's behavioral class
(\texttt{declare-$\langle$bc$\rangle$}); and \emph{target-load} writes it with a
BC-dispatched apply. Each sub-pattern is verified in isolation, so Augmentation is
grain-correct \emph{by construction from its parts}.
\end{itemize}

\noindent Patterns compose exactly as pipelines do---by $\seqp$---and
\emph{the composite of patterns is again a pattern}: a CTA whose capture is the first
sub-pattern's capture, whose apply is the last's, and whose transform is everything
between. The named vocabulary is therefore \emph{closed under composition}
(Augmentation is one instance of that closure), giving a small, pre-verified set of
pipeline shapes from which real pipelines are assembled and, by
Theorem~\ref{thm:synthesis}, inherit grain correctness automatically.

% =============================================================================
% 6. THREE-LAYER CORRECTNESS   (budget ~1.25 pp)
% Drafted from PDD.md §II.5 (lines 3740–4105). First full draft.
% =============================================================================

\section{Three-Layer Correctness}
\label{sec:correctness}

PDD establishes pipeline correctness in three layers (Table~\ref{tab:three-layer}),
\emph{all at design time, at zero cost}---with no access to data. This is the central
claim of the paper, and it is a theorem: a pipeline's correctness is a property of its
\emph{code}, settled before any data flows (Pipeline Correctness, Thm~\ref{thm:pct}),
under the design's assumptions about its inputs. A data engineer controls the pipeline
he builds, not the data a provider sends through it; so the one obligation that
genuinely depends on data---whether the inputs meet those assumed preconditions---is
\emph{data quality}, not code correctness, and it leaves the zero-cost correctness of
the code untouched. The method discharges that obligation as auto-generated validation
queries at the input boundary every typed system keeps against untyped external data.
Code bugs are eliminated by construction, at zero cost; data bugs are caught at that
boundary.

\begin{table}[t]
\centering\small
\caption{All three correctness layers are established \emph{at design time, at zero
cost}---with no access to data. Only data-dependent constraints add a runtime check,
and it validates the \emph{inputs} (data quality), not the code.}
\label{tab:three-layer}
\begin{tabular}{@{}p{2.0cm}p{3.4cm}p{1.8cm}@{}}
\toprule
Layer & Establishes (design time, zero cost) & Runtime input check \\
\midrule
Grain & output grain $=$ target (CalcG) & --- \\
Behavioral class & operations respect their BC (typing) & --- \\
Domain (contract-derivable) & $\subseteq_c$, $\eqg$, cardinality bounds (contracts) & --- \\
Domain (data-dependent) & FK, ranges, business rules, given input preconditions & generated query \\
\bottomrule
\end{tabular}
\end{table}

\subsection{Layer 1: Grain Correctness via CalcG (Zero Cost)}
\label{sec:correctness:calcg}

The grain inference rules of Part~I~\cite{graintheory-pods} assign every PDA
operation a grain transformation that involves only finite set operations
($\tun,\tin,\tdiff,\subt$) on field sets. The grain-computation operator
$\calcg[\textsf{op}](G_{\text{in}})$ applies these rules; for example captures,
applies, set operations, \texttt{enrich}, and \texttt{sel-safe} preserve the input
grain, \texttt{groupby} sets the grain to the grouping columns, and an equi-join
yields $\textsf{JoinGrain}(G_1,G_2,Jk)=G_1\tun(G_2\tdiff Jk)$ (Thm.~37). CalcG
composes functionally over pipelines, $\calcg[\textsf{op}_1\seqc\textsf{op}_2] =
\calcg[\textsf{op}_2]\circ\calcg[\textsf{op}_1]$, which is what makes
Theorem~\ref{thm:synthesis} hold.

A pipeline is \emph{grain-correct} when the data it \emph{produces} lands at
exactly the intended level of detail: for every source, the grain of its output
collection equals the target's declared grain $\grain{\textsf{Target}}$. This is a
semantic property of the pipeline's output---the one a fan trap or a missing
aggregation violates (\S\ref{sec:intro})---yet, by the grain homomorphism
(\S\ref{sec:background}), it is decidable \emph{without running the pipeline}.

\begin{theorem}[Data-independent grain correctness~\cite{graintheory-pods}]
\label{thm:calcg}
A pipeline $\textsf{pip}:\textsf{Pipeline}\,(\Coll{\textsf{Source}})\,(\Coll{\textsf{Target}})$
is grain-correct iff
$\calcg[\textsf{pip}](\grain{\textsf{Source}}) = \grain{\textsf{Target}}$.
The check operates entirely on the type-level schema and its declared
determinations---foreign keys and declared grains, which are \emph{metadata, not
data}---so it needs no data access; CalcG is total and decidable, running in
$O(\lvert V\rvert\cdot k)$ time on a DAG of $\lvert V\rvert$ operations with
fan-in $k$.
\end{theorem}

\paragraph{Running example.} For the customer-order pipeline---order events
(grain $\textsf{CustomerId}\times\textsf{OrderDtm}$) joined to customers and
products, then aggregated to a per-customer summary (grain
$\textsf{CustomerId}$)---CalcG threads the grain through the chain:
\texttt{delta} preserves $\textsf{CustomerId}\times\textsf{OrderDtm}$;
\texttt{lookup} on products preserves it (grain-ordered, no fan-out);
\texttt{groupby-ek} sets it to $\textsf{CustomerId}$; \texttt{upsert} preserves it.
The computed grain $\textsf{CustomerId}$ equals the target's declared grain, so the
design is grain-correct---established before any data is read. Had the designer
used \texttt{lookup-low} (fan-out) or omitted the grouping, CalcG would report a
grain mismatch, the signature of a double-counting or row-multiplication bug.

\subsection{Layer 2: Behavioral-Class Correctness (Compile Time)}

Each operation declares its behavioral-class constraint as a type parameter, and
the type checker enforces it: a \texttt{delta} on an \textsf{IsEntity} collection,
or an \texttt{append} to an \textsf{IsEntity} target, is a type error. Non-canonical
reads must go through the explicit converter pattern (\S\ref{sec:pda:capture}). By
the BC-preservation principle, grain correctness already implies BC preservation
wherever grain is preserved, so this layer adds only the cases where BC is
deliberately changed---and those are caught at compile time, for free.

\subsection{Sound Re-Graining: From Collection Keys to Type Grain}
\label{sec:correctness:regrain}

The deliberate grain/BC changes of Layer~2 are the \emph{declares}
$\texttt{declare-$\langle$bc$\rangle$}$ of \S\ref{sec:pda:declare}: each
\emph{re-lenses} a collection, declaring a new grain $K$ and behavioral class
$\langle\textsf{bc}\rangle$. A re-lens carries no data
operation; it is sound exactly when the data already keys on $K$. Pinning down ``keys
on $K$'' needs a bridge between the collection level, where keys are observed, and the
type level, where grain lives. Throughout, $k:R\twoheadrightarrow K$ is a \emph{field-set
projection} (the double-headed arrow $\twoheadrightarrow$ denotes a surjection: $K$
retains a subset of $R$'s fields, so $k$ is onto $K$). Call such a $k$ a \emph{superkey}
of a collection $c:C\,R$ when it is injective on the elements of $c$, and a \emph{key}
when it is moreover \emph{minimal}---no proper sub-projection of $k$ is still injective.

\begin{theorem}[Universal Superkey is the Grain]
\label{thm:univ-key-grain}
For a field-set projection $k:R\twoheadrightarrow K$ (so $k$ is surjective), the
following are equivalent:
\emph{(1)} $k$ is injective on the \emph{type} $R$;
\emph{(2)} $k$ is a superkey of \emph{every} collection $c:C\,R$;
\emph{(3)} $k$ is itself \emph{grain-preserving}---a grain isomorphism $R\cong K$
\emph{witnessed by $k$}, whence $R\eqg K$. Moreover $K$ is a \emph{grain} of $R$
(irreducible, $\grain{K}=K$) iff no proper sub-projection of $k$ satisfies
\emph{(1)--(3)}---minimality is \emph{type-level}, not per-collection. (Clause~(3)
names $k$ as the witness on purpose: that $R\eqg K$ holds for \emph{some} isomorphism is
weaker---an arbitrary surjection between grain-equivalent types need not be injective; it
is $k$ \emph{itself} being the isomorphism that yields injectivity.)
\end{theorem}

\begin{proof}[Proof sketch]
(1)$\Rightarrow$(2) is immediate: injective on the type $R$ implies injective on the
elements of every collection. (2)$\Rightarrow$(1), contrapositively, turns a type-level
collision $k\,r_1=k\,r_2$ with $r_1\neq r_2$ (whose grains differ by Grain
Uniqueness~\cite{graintheory-pods}) into a two-element collection on which $k$ is no
superkey. (1)$\Leftrightarrow$(3): a surjection is injective iff it is a bijection, so
$k$ injective on $R$ coincides with $k$ being an isomorphism $R\cong K$ (clause~(3)),
whence $R\eqg K$. Minimality is the Grain Inference Theorem's irreducibility
condition~\cite{graintheory-pods}. \fullproof
\end{proof}

\begin{remark}[The quantifier is enforced, not observed]
\label{rem:enforced-quantifier}
Clause~(2) cannot be checked against data: a finite collection can always be filtered
until a proper sub-projection becomes a key (a given collection may key on less than the
grain)---a key is not inferable from any single instance, which is why relational
systems \emph{enforce} keys as declared constraints rather than infer them. Here the quantifier is discharged at the type level, never
empirically---the grain is injective on $R$ by construction (\S\ref{sec:background})---
and, for the re-graining below, only over the collections a preceding operation
actually produces, on which that operation's postcondition \emph{guarantees} the key.
The pipeline is the enforcer; the type-checker is the auditor.
\end{remark}

\begin{corollary}[Design-Time Sound Re-Grain]
\label{cor:sound-regrain}
$\texttt{declare-$\langle$bc$\rangle$}\,k$ re-declares its input at grain $K$ and class
$\langle\textsf{bc}\rangle$ under the \emph{precondition} that $k$ is a key of every
input; by Theorem~\ref{thm:univ-key-grain} this makes $K$ the grain of the re-declared
type $\textit{Res}$. Hence for \emph{any} \pda{} pipeline $t:C\,A\to C\,B$---one operation
or arbitrarily many---whose contract guarantees that $k$ is a key of $t\,c$ for every
$c$, the composition $t\seqc\texttt{declare-$\langle$bc$\rangle$}\,k$ is sound \emph{by
construction}: $t$'s postcondition discharges the precondition on every collection $t$
produces, the type-checker verifies the composition at design time, and the result
denotation is
$(\grain{}=K,\ \BC{}=\langle\textsf{bc}\rangle,\ \ek{}\text{ derived from }K)$---with no
data read.
\end{corollary}

\begin{proof}[Proof sketch]
The precondition makes the input grain-unique at $K$, so by
Theorem~\ref{thm:univ-key-grain} $K$ is its grain; and since the re-lens does no data
action---dropping and choosing nothing---declaring $\grain{}=K$ is sound. $t$'s
postcondition supplies the key on every collection $t$ produces, discharging the
precondition over the image, so the composition type-checks with no premise depending on
data; the result denotation is the one $\texttt{declare-$\langle$bc$\rangle$}$ declares. \fullproof
\end{proof}

\paragraph{Why a \emph{key}, and what $\textit{Res}$ is.} Minimality is mandatory, not a
refinement: a bare superkey gives only grain-\emph{equivalence} ($R\eqg K$ with
$\grain{K}\psub K$), under which the contract's claim $\grain{\textit{Res}}=K$ would be
false. The re-lens performs \emph{no data action}---it neither aggregates nor
selects---so it drops and merges nothing, type-checking exactly when a witness
establishes the precondition. And because no data is touched, $\textit{Res}$ keeps the
input's schema ($\textit{Schema}(\textit{Res})\equiv\textit{Schema}(R)$) yet carries a
different denotation ($\sem{\textit{Res}}\neq\sem{R}$); since identity in denotational
design is \emph{semantic}---$A=B\iff\sem{A}=\sem{B}$~\cite{elliott2009denotational}---%
$\textit{Res}$ is a \emph{distinct} grain-theoretic type, the same structure re-lensed to
a new grain and class, \emph{not} an isomorphic alternative grain of $R$ (the
\textsf{CustomerEntity}-vs-\textsf{CustomerMVersion} distinction of
Part~I~\cite{graintheory-pods}).

\paragraph{The postcondition is all that matters.} Not which operation, nor how long the
pipeline: $\texttt{checkout}$ (one version per entity), $\texttt{pick}$ (one row per
partition key), or any pipeline ending in such a step all supply the key. This is exactly
the $\textsf{IsMultiVersion}\to\textsf{IsEntity}$ re-grain of \S\ref{sec:pda:declare}
($\texttt{checkout}\seqc\texttt{declare-entity}\seqc\texttt{sel-safe}$), where
$\texttt{checkout}$'s postcondition supplies the key and the $\texttt{declare-entity}$
re-lens is the type-level move the corollary justifies. When $k$ is \emph{not} a key of
the input, the re-lens does not apply---collapsing then requires \emph{choosing} a
representative, the job of $\texttt{group-by}$ (which aggregates) or $\texttt{pick}$
(which selects by a declared ordering), not of a re-lens.

\paragraph{Stricter target classes.} Beyond the grain, a stricter target class carries
its own structural invariant (e.g.\ consecutive versions differ, for
$\textsf{IsMultiVersion}$) as a further precondition on the re-lens, discharged the same
way---by a preceding operation's postcondition, at the type level.

\subsection{Layer 3: Domain Correctness}
\label{sec:correctness:domain}

The third layer covers business rules beyond grain and BC. Its constraints are not
opaque predicates: they are drawn primarily from a \emph{structured vocabulary}---the
\emph{collection relations} of Table~\ref{tab:collrel}---in three scopes: \emph{element}
($\forall e\in c.\,P\,e$), \emph{collection} (a predicate over one collection, e.g.\
grain-uniqueness), and \emph{cross-collection} (a predicate relating several, e.g.\ a
foreign key). The constraints that matter most \emph{are} these relations: a foreign key
is a subset ($\subseteq$), and uniqueness is a grain-\emph{key}
($\texttt{IsKeyOf}$---equivalently a superkey; the grain projection is always one,
\S\ref{sec:background}). Their implication hierarchies
($=_c\Rightarrow=_g\Rightarrow=_{ek}$; $\subseteq_c\Rightarrow\subseteq_g\Rightarrow
\subseteq_{ek}$) express strictness gradations a single first-order predicate cannot---the
same FK at three strengths (identical values, up to grain isomorphism, up to entity key).
Arbitrary first-order-logic predicates enter only as an \emph{escape hatch}, for business
rules that reduce to no relation. This choice is not cosmetic---it is the crux of the
layer's cost. Only a relation can be \emph{free}: a property a contract \emph{guarantees}
as a collection relation is discharged at zero cost, at design time, whereas a raw
first-order predicate is \emph{always} specification-time---it generates a boundary query
regardless. The collection relations are therefore not an alternate encoding of the same
FOL predicates; they are precisely the fragment the pipeline's \emph{own contracts}
discharge for free, so stating a constraint relationally is what draws it onto the
zero-cost, code-correctness side of the line rather than the data-quality side
(\S\ref{sec:correctness}). Domain correctness accordingly splits by \emph{cost}.

\paragraph{Contract-derivable constraints (zero cost).} Many collection-level
relations follow from operation contracts alone, by structural induction over the
pipeline---no data required. A pipeline built from captures automatically produces
a subset of its input ($\subseteq_c$, hence $\subseteq_g$ and $\subseteq_{ek}$)
with bounded cardinality; an apply's result differs from its before-state by
exactly its source; grain-type equivalence $\eqg$ follows wherever CalcG shows
grain preservation; and because the grain projection is always a key
(\S\ref{sec:background}), grain-uniqueness ($\texttt{grain}\;\texttt{IsKeyOf}\;c$) is
free as well. These extend the zero-cost guarantee beyond grain and BC at no
specification cost.

\paragraph{Data-dependent constraints (specification + generated query).} What is
left---foreign-key integrity on values, value ranges, partition counts, and
arbitrary business predicates---genuinely depends on the data. PDD specifies these
\emph{primarily through the collection relations}---a foreign key as an asserted subset
$\subseteq$, uniqueness as a grain-key $\texttt{IsKeyOf}$---and, for business rules that
reduce to no relation, as arbitrary first-order-logic predicates (the \emph{escape
hatch}). All are carried as a $(\textsf{Pre}$, $\textsf{Intra}$,
$\textsf{Post})$ triple, type-checked for well-formedness and composed through the
proof-carrying operator $\mathbin{>\!>\!>}$, which discharges each
$\textsf{Post}_n\Rightarrow\textsf{Pre}_{n+1}$ obligation at composition time
(intermediate obligations that are contract-derivable are inherited as axioms,
so the engineer proves only what is genuinely new). This discharge is itself
\emph{design-time}: the $\mathbin{>\!>\!>}$ obligations are type-level, settled
before any data flows, so the pipeline is correct \emph{by construction given its
preconditions}, and what genuinely depends on data is only whether the inputs
\emph{satisfy} the resulting leaf preconditions---a data-quality question. This full
discharge is the \emph{proof-backed} reading: collapsing every intermediate
obligation to a leaf precondition requires actually proving the
$\textsf{Post}_n\Rightarrow\textsf{Pre}_{n+1}$ entailments, which the
dependently-typed realization (PDL) does mechanically
(\S\ref{sec:correctness:spectrum}). The deployed \pddskill{} checker discharges the
structural layers but does not prove these entailments; it emits a verification query
at each node where a data-dependent constraint is declared. Either way the constraint
is not a hole in code correctness but a data-quality control point---at the input
boundary under a proof, at each declaring node otherwise.

\begin{table}[t]
\centering\small
\caption{Collection relations (the complete family). $\texttt{sel-ek}$ is the entity-key
projection, $\textit{Gr}(c)=\{\textit{grain}\,e\mid e\in c\}$ its set of grain values, and
$f:\texttt{GetG}\,c_1\xrightarrow{\sim}\texttt{GetG}\,c_2$ a grain isomorphism witnessing
$\texttt{GetG}\,c_1\eqg\texttt{GetG}\,c_2$; $f[\textit{Gr}(c_1)]=\{f(v)\mid v\in\textit{Gr}(c_1)\}$
is the image of a set under $f$. Because grain is a superkey (the grain projection is
injective; \S\ref{sec:background}), a collection is determined by $\textit{Gr}(\cdot)$,
so $=_g$ may compare grain-value sets rather than elements. $\supseteq_g$ is grain
\emph{generalization}---a partition into specializations. The key predicates take a
field-set projection $k:R\to\text{Key}$ on the \emph{left} and a collection on the right;
$\texttt{IsKeyOf}$ adds minimality, and the grain projection is always a key
(\S\ref{sec:background}). The equality/subset blocks run strongest to weakest:
$=_c\Rightarrow\, =_g\Rightarrow\, =_{ek}$ (and $=_g\Rightarrow\, =_{\#}$),
$\subseteq_c\Rightarrow\subseteq_g\Rightarrow\subseteq_{ek}$.}
\label{tab:collrel}
\begin{tabular}{@{}ll@{}}
\toprule
Relation & Definition \\
\midrule
$c_1 =_c c_2$ & $\forall e\in c_1.\,e\in c_2 \,\wedge\, \forall e\in c_2.\,e\in c_1$ \\
$c_1 =_g c_2$ & $\texttt{GetG}\,c_1\eqg\texttt{GetG}\,c_2 \,\wedge\, f[\textit{Gr}(c_1)]=\textit{Gr}(c_2)$ \\
$c_1 =_{ek} c_2$ & $\texttt{sel-ek}\,c_1 =_g \texttt{sel-ek}\,c_2$ \\
$c_1 =_{\#} c_2$ & $\lvert c_1\rvert = \lvert c_2\rvert$ \\
\addlinespace
$c_1 \subseteq_c c_2$ & $\forall e\in c_1.\; e\in c_2$ \\
$c_1 \subset_c c_2$ & $c_1\subseteq_c c_2 \,\wedge\, \neg(c_2\subseteq_c c_1)$ \\
$c_1 \subseteq_g c_2$ & $\exists P.\; \texttt{filter}\,P\,c_2 =_g c_1$ \\
$c_1 \subseteq_{ek} c_2$ & $\exists P.\; \texttt{filter}\,P\,c_2 =_{ek} c_1$ \\
\addlinespace
$c \supseteq_g \{c_i\}_{i=1}^{n}$ & $\bigl(\bigwedge_i c_i\subseteq_g c\bigr) \,\wedge\, \textstyle\sum_i\lvert c_i\rvert=\lvert c\rvert$ \\
\addlinespace
$k\;\texttt{IsSuperKeyOf}\;c$ & $\forall e_1,e_2\in c.\; k\,e_1=k\,e_2\Rightarrow e_1=e_2$ \\
$k\;\texttt{IsKeyOf}\;c$ & $k\;\texttt{IsSuperKeyOf}\;c \,\wedge\, k\text{ minimal}$ \\
\bottomrule
\end{tabular}
\end{table}

\paragraph{Collection relations and referential integrity.} The constraints above
draw on a small algebra of \emph{collection relations} (Table~\ref{tab:collrel})
grading how two collections correspond---by elements ($=_c$, $\subseteq_c$), by grain
up to isomorphism ($=_g$, $\subseteq_g$: the same data under different schemas, e.g.\
$\textsf{CustomerId}\leftrightarrow\textsf{Username}$), or by entity ($=_{ek}$,
$\subseteq_{ek}$: the same entities at possibly different grains, e.g.\ a
multi-version history and its current state). These grade \emph{referential
integrity}: the foreign key ``$c_1$'s FK column is contained in $c_2$'s grain'' is
\emph{strict} under $\subseteq_c$ (keys match by value), a \emph{grain} FK under
$\subseteq_g$ (match up to representation---a $\textsf{DeptID}$ resolves against a
$\textsf{DeptSlug}$-keyed table), or an \emph{entity-key} FK under $\subseteq_{ek}$
(an event stream's keys resolve against the entity table despite a finer grain). The
hierarchy grades \emph{strictness}; \emph{cost} is a separate axis, set by
provenance. A $\subseteq$ that an operation's contract guarantees---a capture's
output $\subseteq_c$ its input, with the implied $\subseteq_g$/$\subseteq_{ek}$---is
contract-derivable and free at every level. A $\subseteq$ \emph{asserted} between
independent collections---any FK constraint, strict or up-to-isomorphism---is
data-dependent and compiles to the anti-join of
\S\ref{sec:correctness:queries}. Uniqueness ($\texttt{IsKeyOf}$) and
generalization/partition ($\supseteq_g$) grade the same way: free when an operation's
contract establishes the key or partition (a \texttt{groupby} makes its grouping
projection a key, a \texttt{pick} its partition projection), data-dependent---and
compiled to a grouping query---when asserted of independent data.

\subsection{From Specifications to Verification Queries}
\label{sec:correctness:queries}

A well-formed FOL predicate has syntactic structure that compiles mechanically to
a verification query in SQL or PySpark: $\forall r\in c.\,P$ becomes
\texttt{SELECT * FROM c WHERE NOT (P)}, $\textsf{FK}(c_1,f,c_2)$ becomes an
anti-join, $\textsf{Unique}(c,f)$ a \texttt{GROUP BY \dots HAVING COUNT(*)>1}, and
so on. Each query is designed to return the empty set exactly when its predicate
holds.

\begin{theorem}[Compilation correctness]
\label{thm:compile}
For every well-formed predicate $P$ and its compiled query $Q(P)$,
\[
P \text{ holds} \;\Longleftrightarrow\; Q(P) \text{ returns the empty set.}
\]
Because the compilation is structural recursion over the predicate syntax, there
is no manual translation step at which the test can drift from the specification.
\end{theorem}

\noindent This closes the loop. The pipeline architect writes the design and its
constraints once; the design's correctness is checked from structure alone, at zero
cost; and the only runtime artifacts---the data-quality checks the data engineers
actually run---are \emph{generated} from the same specification, guaranteed by
Theorem~\ref{thm:compile} to test exactly what was specified. Putting the layers
together, a pipeline is correct by construction when
$\calcg[\textsf{pip}](\grain{\textsf{Source}})=\grain{\textsf{Target}}$ (grain),
the design is BC-well-typed (BC), and the domain triple holds (domain)---all three
established at design time given the inputs' preconditions, eliminating double
counting, spurious join records, grain mismatches, and BC violations by construction
rather than by testing. The full
constraint language and its mechanization in Agda~PDL are
in-progress~(\S\ref{sec:related}); the FOL$\to$SQL generator is already
implemented in \sysname{} (\S\ref{sec:toolchain}).

The three layers together formalize the paper's organizing principle---\emph{a
pipeline's correctness is a property of its code, settled at design time}---as a
single data-independence theorem.

\begin{theorem}[Pipeline Correctness]
\label{thm:pct}
Let
\[\textsf{pip}:\textsf{Pipeline}\,(\Coll{\textsf{Source}})\,(\Coll{\textsf{Target}})\]
carry a constraint specification $(\textsf{Pre}$, $\textsf{Intra}$, $\textsf{Post})$:
$\textsf{Pre}$ a precondition on the source, $\textsf{Post}$ the target property (its
grain, behavioral class, and domain predicates), and $\textsf{Intra}$ the stateful
condition of any apply. If \emph{(i)}~$\textsf{pip}$ is BC-well-typed,
\emph{(ii)}~$\calcg[\textsf{pip}](\grain{\textsf{Source}})=\grain{\textsf{Target}}$,
and \emph{(iii)}~composition discharges every internal obligation
$\textsf{Post}_i\Rightarrow\textsf{Pre}_{i+1}$, then for every source collection $c$,
\[
c\models\textsf{Pre}\;\Longrightarrow\;\textsf{pip}\,c\models\textsf{Post}\wedge\textsf{Intra},
\]
and this implication is established \emph{with no data access}. The sole
data-dependent obligation is whether $c\models\textsf{Pre}$---a data-quality check at
the input boundary, not a property of the pipeline's code.
\end{theorem}

\begin{proof}[Proof sketch]
The conclusion conjoins the three layers, each discharged statically. \emph{Grain:}
by Theorem~\ref{thm:calcg}, (ii) forces the output grain to equal the target's,
data-independently. \emph{Behavioral class:} (i) is decided by the type checker.
\emph{Domain:} the specification composes by structural induction over the pipeline
DAG---each operation is correct against its own contract,
$\textsf{Pre}_{\textsf{op}}\Rightarrow(\textsf{Post}_{\textsf{op}}\wedge\textsf{Intra}_{\textsf{op}})$,
true by the operation's denotation, and the obligations~(iii) chain these so the leaf
$\textsf{Pre}$ entails the target $\textsf{Post}$; contract-derivable intermediate
obligations are inherited as axioms (\S\ref{sec:correctness:domain}). All three
discharges are type-level, so the entailment holds with no data read. Proved in the
semantic domain, it transfers to the concrete run by the grain homomorphism
(\S\ref{sec:background}), and each domain predicate's runtime check is faithful to its
specification by Theorem~\ref{thm:compile}. \fullproof
\end{proof}

\noindent Theorem~\ref{thm:pct} is the \emph{unified-correctness} predicate of the
Agda~PDL encoding (\S\ref{sec:related}). It is what makes ``settled at design time''
precise: the implication $\textsf{Pre}\Rightarrow\textsf{Post}$ is pure code analysis,
and the only \emph{data} residue is the input-boundary check of
\S\ref{sec:correctness:spectrum}---data quality, not code correctness. The one further
obligation, that the executable faithfully realizes the design, is discharged by the
grain homomorphism (Thm.~\ref{thm:compile}), again without reading data. Note too that
the theorem is stated \emph{against the design's declared} $\textsf{Post}$: it
guarantees the pipeline computes what its specification says, not that the
specification captures intent---the latter is what the human reviews on the design.

\subsection{Prevention, Not Detection}
\label{sec:correctness:prevention}

The three layers do more than \emph{verify} a finished design; they make whole
classes of bug \emph{unrepresentable}---the discipline of making illegal states
unrepresentable~\cite{minsky2011ocaml}, where a property carried by the type system is
not a test that might fail but a state that cannot be written. The mechanism is the
behavioral class: it fixes, at the type level, which reads and writes are well-typed
and at what grain. Temporal reads are gated by the class---\texttt{checkout}, the
current-version read, requires \textsf{IsMultiVersion}---and collapsing a collection to
a coarser grain requires a \emph{key} that only such a read establishes. We make this
concrete with two real bugs, each a costly silent error that ships through tests, and
each the same omission: using a \emph{versioned} (\textsf{IsMultiVersion}) collection
without the \texttt{checkout} its canonical reading demands. One is on the read side,
one on the write side; in both, the wrong design is caught at design time---by
\pddskill{} today, on any pipeline, and ruled out by the type level itself in the
dependently-typed PDL (\S\ref{sec:correctness:spectrum}).

\begin{figure}[t]
\begin{lstlisting}[style=pda]
-- Dim : IsMultiVersion   (grain EK x FromDtm)
-- a current-state lookup needs Dim at entity grain
-- WRONG: re-lens an all-versions read to entity grain
dimEnt = dim |> full |> declare-entity ek
--   declare-entity needs  ek IsKeyOf input;
--   a full read has many versions/entity -> unmet
-- RIGHT: checkout makes ek a key (one current row
--   per entity), discharging the precondition
dimEnt = dim |> checkout t |> declare-entity ek
\end{lstlisting}
\caption{Read side. \texttt{full} reads all versions; presenting the dimension at
entity grain needs the re-lens \texttt{declare-entity}, whose key precondition
($\text{ek}\;\textsf{IsKeyOf}$ input) only \texttt{checkout}---the canonical
\textsf{IsMultiVersion} read---discharges. \pddskill{} flags the missing
\texttt{checkout} at design time; in the dependently-typed PDL the unmet precondition
makes the all-versions design ill-typed.}
\label{fig:read-bug}
\end{figure}

\paragraph{Example 1: an over-broad read.} A pipeline maintains an SCD2 customer
dimension $\textsf{Dim}$ (\textsf{IsMultiVersion}, grain
$\textsf{CustomerKey}\times\textsf{EffectiveDt}$). To find the customers affected by a
changed address---the address carries no customer key---it looks the address up
against $\textsf{Dim}$ and \texttt{UNION}s the matched keys into the upsert delta. The
natural code reads \emph{all} versions of $\textsf{Dim}$, so an address that belonged
to a \emph{past} version of a customer still matches and emits a spurious key; the
target then accrues bogus new versions. (The \texttt{UNION} is a red herring---it is
correct; the emitted key is a distinct wrong \emph{inclusion}, not a duplicate to
collapse.) PDD prevents this by grain: the lookup needs $\textsf{Dim}$ at \emph{entity}
grain (one current row per customer), reached only through the re-lens
\texttt{declare-entity}, whose precondition---that the customer key is a \emph{key} of
its input---holds only after \texttt{checkout} collapses each customer to its current
version. An all-versions read cannot discharge it: \pddskill{} flags the missing
\texttt{checkout} at design time on any pipeline, and in the dependently-typed PDL the
unmet precondition makes the design ill-typed by construction
(Fig.~\ref{fig:read-bug}, \S\ref{sec:correctness:spectrum}). \texttt{checkout} is the
canonical read for \textsf{IsMultiVersion}.

\begin{figure}[t]
\begin{lstlisting}[style=pda]
-- stg: one row per line item; each file re-sends
-- an order's COMPLETE line-item set
-- WRONG: treat each line item as its own event
snap = reads |> pick {lineKey} file DESC
-- latest file PER LINE ITEM; a dropped line is
-- still "latest for itself" => never deleted
-- RIGHT: the file versions the ORDER
snap = reads
   |> groupby-multiversion {orderKey} {file}
                           {lines = collect(*)}
   |> checkout now      -- latest file per order
   >> explode lines     -- back to line-item rows
\end{lstlisting}
\caption{Write side. Keeping the latest file \emph{per line item} (top) never deletes a
dropped line item. Typing the input as \textsf{IsMultiVersion} at the \emph{order}
grain forces an order-level \texttt{checkout} (bottom) that drops it by construction.}
\label{fig:write-bug}
\end{figure}

\paragraph{Example 2: a missed deletion.} A source re-sends, in every file, an order's
\emph{complete} set of line items; a line item removed upstream simply stops appearing,
with no delete record. Flattened to one row per line item, the tempting design keeps
the latest file \emph{per line item}---but a dropped line item is still the latest
\emph{for itself}, so it is never deleted (Fig.~\ref{fig:write-bug}, top). The fix is
to state the input's true semantics: each file is a new version of the \emph{order},
not the line item. Typed that way, the algebra reconstructs the order version
(\texttt{groupby-multiversion}), \texttt{checkout}s the latest file \emph{per order},
and \texttt{explode}s back to line-item rows---dropping the deleted one by construction
(Fig.~\ref{fig:write-bug}, bottom). This sharpens the principle: a behavioral-class
\emph{label} is necessary but not sufficient; declaring the input \textsf{IsEvent}
type-checks yet still admits the wrong \texttt{pick}. Only the \emph{right} class and
grain---\textsf{IsMultiVersion} at the order grain---forces the correct
\texttt{checkout}.

\paragraph{Prevention versus the detection lottery.} Each bug is invisible to today's
defenses. It produces a wrong result only on a specific data pattern---a
historical-only address, a line item dropped between two files---that sampled tests
rarely contain, and a runtime check catches it only if someone already suspected the
versioning and wrote exactly that assertion. Prevention depends on neither
coincidence---and the alternative is demonstrably unreliable. Under the code-first
status quo both bugs were committed by experienced data engineers working without a
denotational design; and, given the same inputs (the natural-language specification
and schemas) but no design, both were reproduced by a frontier model, Anthropic's
Opus~4.8 (1M-context). Neither human expertise nor model capability substitutes for
the design step. Table~\ref{tab:devmodes} contrasts the development modes: ad-hoc
coding, a natural-language spec implemented by a human or an LLM, and black-box
agentic codegen all \emph{can} ship the bug, whereas a typed design makes it
unrepresentable.

\begin{table}[t]
\centering\small
\caption{The same versioning bug under four development modes. Only a typed design
\emph{prevents} it; the others can at best \emph{detect} it, and only with the exact
triggering data \emph{and} the exact verification query.}
\label{tab:devmodes}
\begin{tabular}{@{}L{2.4cm}L{2.7cm}c@{}}
\toprule
Development mode & Artifact the engineer controls & Prevented? \\
\midrule
jump-in-and-code & code only & no \\
NL spec $\to$ implement (human or LLM) & prose, silent on version grain & no \\
black-box agentic codegen & nothing reviewable & no \\
\textbf{\pdd{} (typed design)} & typed design (\textsf{IsMultiVersion}) & \textbf{yes} \\
\bottomrule
\end{tabular}
\end{table}

\subsection{The Verification Spectrum: Two Realizations of the Checker}
\label{sec:correctness:spectrum}

The three layers describe \emph{what} must hold; they do not fix \emph{how strongly}
it is established. That is a separate axis---software engineering's
\emph{tests--types--proofs} spectrum (\S\ref{sec:related})---and the same \pdd{}
design can be checked at more than one point on it, because the ``checker'' of
\S\ref{sec:intro} is a \emph{role}, not a fixed tool. \Sysname{} realizes that
role two ways (Table~\ref{tab:spectrum}), and the choice is monotone: \emph{the
richer the type system, the less is left for runtime}.

\begin{table}[t]
\centering\small
\caption{Two realizations of the checker on one \pdd{} design, against the
runtime-test baseline. The residue shrinks monotonically as assurance strengthens:
structural obligations fall to design-time typing, data-dependent obligations fall to
a machine-checked proof, and only the untyped-input boundary is irreducible.}
\label{tab:spectrum}
\begin{tabular}{@{}L{1.7cm}L{2.6cm}L{2.6cm}@{}}
\toprule
Checker & Code-correctness assurance & Runtime residue \\
\midrule
tests (ad-hoc $+$ DQ tools) & none---run and inspect outputs & a query per node, for everything \\
\textbf{\pddskill{}} (types; deployed) & grain $+$ BC $+$ contract-derivable, discharged at design time as AI-reviewed evidence & data-dependent checks (FK, ranges, business rules), per node \\
\textbf{PDL} (proofs; Agda/Lean~4) & \emph{all} layers, business predicates included, machine-checked & input boundary \emph{only} \\
\bottomrule
\end{tabular}
\end{table}

At the runtime-test end, every obligation is a query against materialized data---one
or more per DAG node, in the manner of dbt's per-node tests~\cite{dbt}. The \pddskill{}
checker moves the \emph{structural} obligations---grain, behavioral class, and the
contract-derivable constraints---to design time, discharging them into an evidence
document and leaving only the genuinely data-dependent constraints (FK integrity,
value ranges, business predicates) as generated queries
(\S\ref{sec:correctness:queries}). At the dependent-types end, PDL---the Pipeline
Design Language, an embedding of \pda{} and its constraint logic in Agda/Lean~4---lifts
even those predicates to \emph{types}: a value violating a precondition cannot be
constructed, so a design that type-checks carries a machine-checked proof that its
output meets grain, behavioral class, \emph{and} the business predicates. The
type-checker \emph{is} the proof, discharged by recompiling (\texttt{agda~--safe}).

Two points sharpen what the proof end buys. First, the residue never vanishes: a
deployment engine (Databricks, Snowflake, BigQuery) ingests \emph{untyped} external
data, so the typed preconditions cannot be enforced \emph{on the inputs}; they must
be checked, which is the input-boundary query the executable design emits (PDLC
Step~3, \S\ref{sec:overview:pdlc}). A proof establishes ``\emph{if} the inputs meet
their preconditions \emph{then} every downstream collection is correct''; the
boundary query checks the antecedent. Second, the proof is what \emph{collapses the
residue to that boundary}. For a \emph{deterministic} design this backward propagation
is always available: a target predicate $P$ on the output is, verbatim, a predicate on
the inputs---namely $P$ precomposed with the pipeline's denotation
$\sem{\textsf{pip}}$---so the proof-carrying composition
(\S\ref{sec:correctness:domain}) rewrites each per-node constraint into a condition at
the input. The per-node data-quality queries the deployed checker emits---potentially
hundreds across a large DAG, one or more per transformation output---are then entailed
by the boundary check and become \emph{formally redundant}: with the inputs verified,
every intermediate and target node is correct by the proof. Two honest caveats. What
collapses is the residue's \emph{location}, not always its \emph{cost}: for an
\emph{emergent} constraint---one with no local pre-image on the source, e.g.\ a
per-group sum bound---the boundary predicate \emph{is} that aggregation and is no
cheaper to check than recomputing the pipeline, a cost the deployed checker instead
pays per node. And one obligation is not data quality at all: that the executable
\emph{faithfully realizes} the design under the non-determinism the semantic domain
abstracts away---evaluation order, non-associative floating-point reduction---which the
grain homomorphism (\S\ref{sec:background}) and compilation-correctness
(Thm.~\ref{thm:compile}) govern, not a boundary query. Within those bounds PDL does
more than raise assurance from reviewed to machine-checked; it shrinks the runtime
footprint from a check at every node to a single check at the boundary.

Across both realizations the artifact the human consumes is a \emph{certificate}---a
machine-checked proof (PDL) or an evidence ledger (the deployed checker)---so verifying
a design reduces to \emph{checking} that certificate, a cheap and delegable act.
Validating the \emph{specification} itself---that its target denotation and business
predicates capture intent---remains the one irreducibly human task
(\S\ref{sec:intro}); it is where a proof, silent on the property one forgot to state,
cannot help.

% =============================================================================
% 7. THE PDD TOOLCHAIN: A PRODUCTION REALIZATION   (budget ~1.0 pp)
% Anonymous by default: the software is "the PDD toolchain" (\edbname + \pddskill +
% \buildskill). Branded names, if ever enabled, come from name macros defined in the
% private branding file (not shipped); \brandednote carries any branded-only aside.
% Production figures are non-identifying CONTEXT, never measured results (see \S8).
% NOTE: arXiv publishes .tex source -> this file must name no company or system.
% =============================================================================

\section{\Sysname}
\label{sec:toolchain}

The methodology of \S\S\ref{sec:pda}--\ref{sec:correctness} is realized end-to-end
by \sysname{}, a design-first system in production use.\brandednote{} It has three
cooperating components plus the human workflow that ties them together.

\paragraph{Design and verification.} The design tool (a
Claude skill, \pddskill{}) takes the source and target data semantics
($\grain{\cdot}$, $\ek{\cdot}$, $\BC{\cdot}$) and produces a \emph{PDD design}---a
formal specification of the pipeline, expressed in the PDA operator algebra. It discharges grain correctness with CalcG, BC
correctness by typing, and the contract-derivable domain constraints
automatically; it emits the remaining data-dependent obligations as FOL predicates
compiled to SQL verification queries (Thm.~\ref{thm:compile}); and it produces a
human-readable \emph{evidence document} recording every CalcG step, BC check,
derived constraint, and generated query. It is the source of this paper's
zero-cost verification and auto-generated queries.

\paragraph{\edbname{} --- the executable target.} The \edbfull{} (\edbname) is a typed
PySpark framework of $90{+}$ reusable operations implementing the PDA catalog
(\texttt{delta}, \texttt{cdc}, \texttt{checkout}, \texttt{augment},
\texttt{lookup}, \texttt{sel-*}, \texttt{append}/\texttt{upsert}/\texttt{commit-upsert},
SCD2, hashing, validation). Because each \edbname{} operation carries its PDA contract, a
correct design maps to correct code one operation at a time, and grain/BC
correctness ``falls out'' of one-line typed calls. This 1:1 PDA$\to$\edbname{}
correspondence is what makes the \edbname{} program an \emph{executable design}
(\S\ref{sec:overview:pdlc})---an implementation whose meaning equals the design's by
the homomorphism (\S\ref{sec:background})---and shows that PDA is realizable on a
production engine.

\paragraph{Code generation --- design to deployment.} A family of Claude skills
(\ddlskill{}, \buildskill{}, \jobskill{})
compiles a PDD design into the executable design: DDL, a faithful \edbname{} implementation
plus its validation queries, and deployable job bundles (PDLC Step~3). The human
verifies the \emph{design}; the agent generates the \emph{code}; the generated
queries verify the data. This is the answer to ``who verifies AI-generated
pipelines?''---the design is the artifact a human can check and an agent can
faithfully compile. Operational tuning for performance, cost, and scale (PDLC
Step~4) is then applied as semantics-preserving rewrites
(\S\ref{sec:synthesis}), preserving the guarantees by construction.

\paragraph{Deployment context and reproducibility.} \Sysname{} is in production use
on an enterprise data-engineering program---not a proof of concept---where the
design-first workflow builds a portfolio of pipelines across a three-layer
(landing / integration / serving) architecture over many heterogeneous source systems. We
treat deployment figures purely as \emph{context}, never as measured results:
reported build-step speedups, for instance, are bounded end-to-end by the non-coding
share of delivery (requirements analysis, data access, review, governance), so a large
\emph{local} speedup is a modest \emph{end-to-end} one---the measured results are the
controlled study of \S\ref{sec:eval}. Because the executable-design framework is
proprietary, the reproducible artifact accompanying this paper is the method's
\emph{outputs}---formal designs, evidence documents, generated verification queries,
and the PDL/Agda mechanization---rather than the engine itself.

% =============================================================================
% 8. EXPERIMENTAL EVALUATION   (budget ~2.25 pp)   [empirical core]
% Experimental design from roadmap §5.2 + exec-deck (anonymized) + pdd-designs.
% Three arms: A (with PDD, deployed checker) / B (with PDD, PDL proofs) /
%   C (without PDD, agentic). C vs A isolates the design step. Anonymous names.
% RESULTS PENDING: measurements scheduled per roadmap timeline.
% Tables below are scaffolded with headers/captions; numbers are placeholders.
% =============================================================================

\ifresults
\section{Experimental Evaluation}
\else
\section{Evaluation Design}
\fi
\label{sec:eval}

\ifresults\else
\noindent\emph{This preprint fixes the methodology and the evaluation design; the
measurement campaign is in progress and the quantitative results will appear in a
later revision. This section states the central claim under test, the research
questions, and the three-arm protocol.}\par\medskip
\fi
We evaluate the central claim of the paper---\emph{most pipeline correctness is
dischargeable at design time at zero cost, leaving only data-dependent obligations
as auto-generated verification queries}---on real production pipelines built with
\sysname{}. The evaluation asks whether the \emph{methodology} delivers on three
fronts: that it is \emph{broad} (multiple patterns, behavioral classes, and modeling
paradigms, not grain alone), that it \emph{prevents} the defects agentic coding ships,
and that its benefits are \emph{measured} (defect capture, discharge, effort, and
faithfulness). The unit of analysis is a \emph{pipeline unit}---one target table built
by one pipeline.
% NOTE: numbers in this section are PLACEHOLDERS pending the measurement campaign
% (roadmap §5.2). Prose states the design and expected result shape.

\subsection{Research Questions}

\begin{description}[leftmargin=2.4em,style=nextline,itemsep=2pt]
\item[RQ1] Does design-first catch correctness defects at \emph{design time} that
ad-hoc and LLM-generated pipelines ship to runtime or production?
\item[RQ2] What fraction of a pipeline's correctness obligations is discharged at
\emph{zero cost} (CalcG grain $+$ BC typing $+$ contract-derivable constraints)
versus requiring a data-dependent runtime check?
\item[RQ3] How many verification queries are auto-generated per pipeline, and what
classes of real defects do they catch on production data?
\item[RQ4] Is code generation \emph{faithful}---does the generated \edbname{} code
preserve every correctness property asserted by the design?
\item[RQ5] How \emph{uniform} are the resulting pipelines---what fraction of each
pipeline's operations is drawn from the shared pattern vocabulary, and how small is
the structural distance between pipelines of the same pattern?
\item[RQ6] What are the developer-effort and code-size deltas versus ad-hoc and
LLM-without-design baselines?
\item[RQ7] What does a fully mechanized PDL design establish beyond the deployed
\pddskill{} checker, and by how much does the proof \emph{collapse the runtime
residue} (per-node data-quality checks $\to$ a single input-boundary check)?
\end{description}

\subsection{Case Studies, Arms, and Protocol}

\paragraph{Case studies.} We use five sanitized production case studies chosen to
span three diversity axes---pipeline pattern, target behavioral class, and modeling
paradigm (Table~\ref{tab:cases})---demonstrating breadth across the space where
correctness matters, not grain alone. Four already exist as PDD designs with evidence
documents (\textsc{Landing-Vault}, \textsc{Dim2-Order}, \textsc{Ft-Order-Rev.},
\textsc{Customer-Summary}); a pure Ingestion case completes the coverage matrix.

\paragraph{Three arms, two reference points.} We build each pipeline unit three ways
with an agent. \textbf{Arm~A (with PDD, deployed):} the design is produced and checked
by the \pddskill{} tool, which emits an evidence document over grain, BC, and
contract-derivable obligations; we report Arm~A for all five cases. \textbf{Arm~B
(with PDD, proved):} the design is mechanized in PDL (Agda), with grain, behavioral
class, and one business predicate proved and \emph{no postulate on the critical path},
so that compilation \emph{is} verification; we anchor Arm~B on the running-example
pipeline (\textsc{Customer-Summary}). \textbf{Arm~C (without PDD):} the same LLM
generates the pipeline directly from natural-language requirements, with \emph{no}
formal design step. The contrast \textbf{C vs.\ A} isolates the contribution of the
\emph{design} from the LLM's coding ability---the central question of the paper---while
\textbf{B vs.\ A} measures what the proof end of the verification spectrum
(\S\ref{sec:correctness:spectrum}) adds over the deployed checker. Two further
\emph{reference points} locate these on the wider landscape: ad-hoc hand-written
PySpark/SQL (the pre-PDD ``before'' state) and runtime-only data-quality tooling
(Great~Expectations/Deequ). A capability matrix (Table~\ref{tab:capability}) shows that
grain and BC defects are \emph{inexpressible} in runtime DQ tools, independent of any
counts.

\paragraph{Ground truth and protocol.} Defects are established two ways: a
controlled \emph{seeded-fault} catalog drawn from a fixed taxonomy
(grain-mismatch, BC-violation, fan/chasm-trap, missing-provider, non-unique
join-key, FK-orphan, NULL-key, SCD2 temporal-overlap)---several seeded from
\emph{real} documented design-time catches, including the FK-augmentation fault
classes catalogued in the case material (\texttt{fk-augmentation-pattern})---and
\emph{historical bug mining} of the pre-PDD tracker for the same domains. Each case is
built under all arms; defects are graded blind against the taxonomy by a rater
unaware of which arm produced the artifact. We instrument the toolchain to emit a
machine-readable \emph{obligation ledger} (RQ2)---one row per CalcG step, BC check,
contract-derivable relation, and data-dependent predicate, already latent in the
evidence documents---tag each generated query by constraint type (RQ3), and log the
design-op$\to$\edbname{}-call trace plus a differential test (RQ4).

\begin{table}[t]
\centering\footnotesize
\setlength{\tabcolsep}{4pt}
\caption{Case-study portfolio. Coverage spans patterns
\{Ingestion, Staging-to-Landing, Augmentation, Aggregation\}, behavioral classes
\{Entity, Event, MultiVersion, Snapshot\}, and paradigms \{Kimball, Data Vault,
source-native\}.}
\label{tab:cases}
\begin{tabular}{@{}llll@{}}
\toprule
Case study & Pattern & Target BC & Paradigm \\
\midrule
\textsc{Landing-Vault} & Staging-to-Landing & Entity\,$+$\,Event & Data Vault \\
\textsc{Dim2-Order} & Augmentation & MultiVer.\ (SCD2) & Kimball dim \\
\textsc{Ft-Order-Rev.} & Augmentation & Event (fact) & Kimball/DV \\
\textsc{Customer-Summary} & Aggregation & Entity & Kimball agg. \\
\textsc{Ingestion} (new) & Ingestion & MultiVer./Snap. & source-native \\
\bottomrule
\end{tabular}
\end{table}

\ifresults
\else
\subsection{Planned Analyses}

The measurement campaign is in progress; per research question we will report:
\textbf{RQ1}~defect capture by phase (design-time vs.\ escape to runtime/production)
across arms, expecting PDD's mass at the design phase and the no-design arm's at
runtime; \textbf{RQ2}~the per-layer obligation ledger (grain / BC / contract-derivable
/ data-dependent) and the zero-cost ratio---counts already extractable from the
existing evidence documents; \textbf{RQ3}~generated-query yield by constraint type and
confirmed defects on production snapshots; \textbf{RQ4}~code-generation faithfulness via
the design-op$\to$\edbname{}-call trace and differential testing; \textbf{RQ5}~pipeline
uniformity---share of operations drawn from the shared pattern vocabulary and
structural edit-distance within a pattern; \textbf{RQ6}~developer-effort and code-size
deltas against the ad-hoc and no-design arms; and \textbf{RQ7}~the proof-anchor
residue-collapse for \textsc{Customer-Summary} (per-node data-dependent checks $\to$ a
single input-boundary check). The capability matrix below (Table~\ref{tab:capability})
is structural and independent of any counts: grain and behavioral-class defects are
\emph{inexpressible} in runtime data-quality tooling.
\fi

\begin{table}[t]
\centering\small
\caption{Defect-class capability matrix: where each approach catches each defect
class. \textbf{D}~$=$~design-time, R~$=$~runtime, ``--''~$=$~inexpressible.}
\label{tab:capability}
\begin{tabular}{@{}lcccc@{}}
\toprule
Defect class & PDD & GE/Deequ & ad-hoc & LLM \\
\midrule
Grain mismatch & \textbf{D} & -- & R & R \\
BC violation & \textbf{D} & -- & R & R \\
Fan/chasm trap & \textbf{D} & -- & R & R \\
FK orphan & \textbf{D}/R & R & R & R \\
Value range / NotNull & \textbf{D}/R & R & R & R \\
\bottomrule
\end{tabular}
\end{table}

\subsection{Threats to Validity}

The case studies come from one enterprise domain; we mitigate with breadth across
patterns, behavioral classes, and paradigms, and argue the patterns are
domain-independent. The no-design arm (C) isolates the design's contribution from the
LLM's coding ability. Blind grading, an independently defined seeded-fault catalog,
and historical-bug cross-checks mitigate author-built bias. Faithfulness is checked
by differential testing in addition to static tracing. Effort figures are reported
only from measured timestamps and LOC; the deployment numbers of
\S\ref{sec:toolchain} are clearly labeled as context. A sanitized artifact (designs,
evidence documents, generated queries, obligation ledgers, and the seeded-fault
catalog) accompanies the paper for reproducibility.

% =============================================================================
% 9. BEYOND PIPELINES: CORRECT-BY-CONSTRUCTION QUERIES   (budget ~0.6 pp)
% Broader-impact argument: PDD's design-time grain/BC verification is not
% pipeline-specific -- a query is the one-node case -- so it is the correctness
% layer the ontology / semantic-layer / text-to-SQL stack lacks. Two industry-
% AGNOSTIC examples (fan trap, all-versions read). Theory applies now; query-layer
% tooling + evaluation are future work (no over-claim). Anonymous, domain-neutral.
% =============================================================================

\section{Beyond Pipelines: Correct-by-Construction Queries}
\label{sec:beyond}

Nothing in PDD is specific to \emph{pipelines}. A pipeline is a DAG of grain-bearing
transformations; a \emph{query} is a single one---the one-node case. \calcg{} verifies
a whole DAG at design time (\S\ref{sec:correctness}); \emph{a fortiori} it verifies one
query---indeed the fan trap that opened this paper (\S\ref{sec:intro}) \emph{is} a
query. So the correctness the method establishes---output grain, behavioral-class-%
respecting reads, no spurious fan-out---applies unchanged to the queries an analyst, a
BI tool, or an AI agent issues \emph{on top of} the data assets a pipeline builds.

This has become urgent. The AI-enabled data stack now being assembled across the
industry---knowledge graphs and \emph{ontologies}, \emph{semantic layers}, and
\emph{natural-language-to-SQL} agents---turns a business question into a query
automatically. These layers solve a real problem: they tell an agent \emph{where} to
look and \emph{how} to phrase a query in the organization's vocabulary. But that is a
different question from correctness. A semantic layer certifies that a query is
\emph{well-formed against the model}; it does not certify that the query
\emph{computes the intended number}. An agent navigating an ontology can emit SQL that
is syntactically valid and semantically plausible yet \emph{logically wrong}, and---run
against live data---return a confident wrong answer with no error raised. The guarantee
the industry is racing to deploy is one of \emph{fluency}, not \emph{correctness}.

The failure modes are the same grain and behavioral-class errors PDD prevents in
pipelines, now on the read path.
\emph{(1) The fan trap.} Asked for total revenue alongside total units, an agent joins
a revenue fact (one row per customer per order line) to a units fact (one row per
customer per shipment) on customer id. The grains are incomparable, so the join
cross-multiplies over the unmatched components: a customer with three order lines and
four shipments yields twelve rows---revenue summed four times, units three. The result
lands at the requested (customer) grain and looks right; both totals are silently
inflated.
\emph{(2) The all-versions read.} A dimension whose attributes change over time---a
customer's segment, a product's list price---is correctly modeled as versioned records
with effective dates (an \textsf{IsMultiVersion}/SCD2 table) that must be read at a
point in time (a \texttt{checkout}). An agent that does not know this reads \emph{every}
version: an entity revised mid-period appears more than once, its measures
double-counted; join that table to a fact and each fact row matches every historical
version, inflating counts and sums in lockstep.
Both are invisible to schema validation \emph{and} to the semantic layer, because
neither is a schema violation---each is a \emph{grain}/\emph{behavioral-class}
violation, exactly the layer of meaning a schema omits (\S\ref{sec:background}).

Correct-by-construction \emph{at design time} is the guarantee these architectures
lack. A query's grain and behavioral-class correctness are decided from the schema
denotations and the query's structure alone (\S\ref{sec:correctness}), so they can be
settled \emph{before the query runs}---the moment it is generated, not after a human
notices the numbers are too large. The same \calcg{} that types a pipeline types a
query; the same behavioral-class discipline that forbids an unchecked versioned read in
a pipeline forbids it in a query. PDD is, in this light, not merely a way to build
correct data \emph{pipelines} but a candidate \emph{correctness layer} for the
AI-enabled data architecture end to end---from the pipelines that build the data assets
to the AI-generated queries that answer business questions over them. Realizing that
layer---extending the design tool and its verification to the query/semantic-layer
boundary, and evaluating it---is ongoing work; the point here is that it demands
\emph{no new theory}, only the grain and behavioral-class denotations this paper
already computes.

% =============================================================================
% 9. RELATED WORK   (budget ~0.75 pp)
% Reuse/adapt from PODS related-work where overlapping; add SYSTEMS angle.
% -----------------------------------------------------------------------------
% Threads:
%  - Functional dependencies & granularity (grain vs FDs) — cite PODS, Codd, Armstrong.
%  - Type-safe / embedded query construction (LINQ, Quill, Slick, type-safe SQL).
%  - Dataflow / pipeline frameworks (Spark, dbt, Airflow) — no correctness guarantees.
%  - Data quality & constraint checking (Deequ, Great Expectations) — RUNTIME, data-dependent;
%    contrast with our design-time, zero-cost discharge.
%  - Denotational design / correct-by-construction (Hughes, Elliott); program synthesis.
%  - Verification of data transformations / provenance.
%  - Position: we UNIFY a typed algebra (design-time, zero-cost) with generated runtime checks
%    only for the residual data-dependent layer.
% =============================================================================

\section{Related Work}
\label{sec:related}

\paragraph{A spectrum: tests, types, proofs.} Software engineering assures
correctness along a spectrum---runtime tests, static type checks, and
machine-checked proofs~\cite{demoura2021lean4}. Its type--proof end is one lineage:
a Hoare triple $\{P\}\,C\,\{Q\}$~\cite{hoare1969axiomatic} states a program's
correctness relative to a specification and discharges it compositionally; a
\emph{refinement type}~\cite{rondon2008liquid} folds the precondition into the type
and has an SMT solver discharge it automatically, trading expressiveness for
push-button checking; a \emph{dependent type} (Agda, Lean~4~\cite{demoura2021lean4},
F*~\cite{swamy2016fstar}, which spans both) lifts the restriction and asks the
programmer for a proof term, trading automation for full expressiveness. The axis is
thus not only how \emph{strongly} correctness is established but at what
\emph{automation--expressiveness} trade-off. Data engineering today sits
almost entirely at the runtime-test end. Dataflow and orchestration frameworks
(Spark~\cite{zaharia2016apache}, Airflow~\cite{airflow}, dbt), and streaming dataflow
models~\cite{akidau2015dataflow,armbrust2018structured,carbone2015flink}, execute
pipelines but make no correctness guarantee about the transformation itself;
data-quality tools (Great~Expectations~\cite{greatexpectations},
Deequ~\cite{schelter2018automating}, dbt tests~\cite{dbt}) check assertions
\emph{against materialized data at runtime}. Because a schema
captures only structure, these runtime checks are the only handle such systems
have---and, run on sampled or partial data, they miss precisely the grain errors
that silently corrupt aggregates (\S\ref{sec:intro}); even property-based
testing~\cite{maciver2019hypothesis} explores values, not grain semantics.
Type-safe query construction (LINQ~\cite{meijer2006linq}, Quill, Slick) pushes
\emph{some} checks to compile time, but types over a relational schema still carry no
notion of grain, entity, or behavioral class, so the errors that matter remain
invisible to them.

\paragraph{Semantic layers and AI-generated queries.} Metrics and semantic
layers---dbt's Semantic Layer and MetricFlow~\cite{dbtsemanticlayer}, and
Malloy~\cite{malloy}---let teams declare dimensions, measures, and the grain of a
metric, and are the closest industrial kin to ``designing in a semantic domain.'' But
they model grain in order to \emph{generate} consistent queries, not to \emph{verify}
a transformation: the declared grain is trusted, never checked against what a query
actually computes, and the check is neither data-independent nor decidable. In
parallel, LLM text-to-SQL~\cite{yu2018spider,li2023bird} turns a natural-language
question into a query directly; benchmarks reward \emph{execution accuracy} on
examples, not a guarantee that a generated query is grain-correct on all inputs. Both
lines answer \emph{where} and \emph{how} to query; neither answers \emph{whether} the
query is correct---the grain and behavioral-class check PDD supplies at design time,
and the reason the method extends from pipelines to the queries that run over these
layers (\S\ref{sec:beyond}).

\paragraph{Foundations: dependencies, categories, provenance.} Grain ordering is the
type-level analogue of a functional dependency---a surjection between grains lifts the
determinacy that Armstrong's axioms~\cite{armstrong1974dependency} give for attributes
from columns to whole types (\S\ref{sec:background}). Category-theoretic models of
databases---functorial data migration~\cite{spivak2012functorial}, algebraic
databases~\cite{schultz2017algebraic}, and attributed
$C$-sets~\cite{patterson2022categorical}---treat schemas and instances as functors and
give transformations a compositional semantics; grain theory is complementary and
deliberately lighter, fixing a single \emph{computable} invariant---the denotation
$(\grain{R},\ek{R},\BC{R})$---rather than a full categorical model, so inference stays
a finite field-set computation. Data provenance~\cite{green2007provenance} tracks
\emph{where} a value came from; grain tracks the orthogonal axis of \emph{at what level
of detail} it is represented---the axis on which double counting and fan traps live.
None of these lines makes pipeline correctness a design-time, type-level property; that
is what PDD contributes.

\paragraph{Our position.} PDD moves pipeline correctness leftward on that
spectrum---from tests toward types and proofs---by making data semantics
\emph{typeable}: in the manner of denotational design~\cite{elliott2009denotational},
a type carries its denotation $(\grain{R},\ek{R},\BC{R})$~\cite{graintheory-pods},
and operations carry contracts propagated through composition---a Hoare-style
$\{\textsf{Pre}\}\,\textsf{pip}\,\{\textsf{Post}\}$ discipline~\cite{hoare1969axiomatic}
over collections, in which grain, behavioral-class, and contract-derivable correctness
are decided before any data is read. Crucially, PDD does not commit to a single point
on the spectrum: it instantiates the correctness \emph{once} and lets a team choose
\emph{how strongly} to discharge it by the pipeline's criticality, complexity, and
scale (\S\ref{sec:correctness:spectrum}). The deployed \pddskill{} checker occupies
the automated, refinement-type-like middle~\cite{rondon2008liquid}: grain inference
(CalcG) and behavioral-class typing are \emph{decidable, domain-specific decision
procedures}---playing the role SMT plays for a refinement-type checker, but purpose-built
for grain, so the structural layers need no general solver at all. We do not leave the
proof end aspirational either: \pda{} and its constraint logic are embedded in
Agda/Lean~4 as the Pipeline Design Language (PDL), a dependently-typed realization in
which we mechanize a complete case-study pipeline end to end (\S\ref{sec:eval}), so the
proof column of the spectrum is demonstrated, not merely cited; the full operation
catalog is in progress. Because
the denotation and \calcg{} are computed from schemas alone, the guarantee is
\emph{universal over the operation set and carrier} (PDD Universality,
\S\ref{sec:background}): the same design-time verification applies to a pipeline
written in relational algebra, in \pda{}, or in the dataflow model's windowing
operators~\cite{akidau2015dataflow}---over a finite table or an unbounded stream alike.
PDD is thus a verification layer atop \emph{any} grain-inferring algebra, where the
frameworks above execute pipelines but verify none. Crucially, this
\emph{separates} the two concerns that runtime testing conflates: \emph{code}
correctness becomes a design-time, type-level property (the Pipeline Correctness
Theorem, \S\ref{sec:correctness}), while runtime checks are confined to \emph{data
quality}---validating that inputs meet the design's assumed preconditions. \pda{} brings to data pipelines what Codd's
algebra~\cite{codd1970relational} brought to ad-hoc data manipulation: a typed algebra
in which correctness is structural, not a test result; PDD is the denotational-design
discipline that composes it over the semantic domain. The resulting division of labor
is \emph{proof-carrying}~\cite{necula1997pcc}: the untrusted generator ships a checkable
certificate---a machine-checked proof or an evidence ledger---and trust reduces to
validating the specification and checking the certificate, never re-reading generated
code at a scale no human can review.

% =============================================================================
% 10. CONCLUSION   (budget ~0.25 pp)
% -----------------------------------------------------------------------------
% Beats:
%  - Recap: PDA gives correct-by-construction pipelines; three-layer correctness; zero cost.
%  - The production toolchain shows it works on real pipelines.
%  - Answer to "who verifies AI-generated pipelines": the human verifies the DESIGN; the
%    algebra + generated queries verify the rest.
%  - Future: full PDL mechanization (Agda/Lean4 — proved CompilationCorrect); broader patterns.
% =============================================================================

\section{Conclusion}
\label{sec:conclusion}

PDD answers ``who verifies AI-generated pipelines?'' by moving the verdict
to design time: \pda{} makes well-typed composition grain-correct by construction, and
the three-layer framework settles code correctness before any data flows, at zero
cost---the \emph{Pipeline Correctness Theorem} (\S\ref{sec:correctness}): the
structural layers are discharged with no data access and a proof-carrying composition
collapses what remains to a single input-boundary check, so a pipeline is correct by
construction \emph{against its specification}, the only residue being whether real
inputs meet that precondition.
Business rules enter as a fourth denotation, \emph{data constraints}---first-order-logic
predicates that propagate through the pipeline and compile to the only check that
survives to runtime, an input-boundary data-quality test, generated rather than
hand-written. Correctness by construction is thus a \emph{dial}, not a fixed point: the same design
spans the tests--types--proofs spectrum (\S\ref{sec:correctness:spectrum})---from
runtime queries, through the deployed \pddskill{} checker (grain and behavioral class
as decidable, domain-specific decision procedures in place of a general SMT solver), to
machine-checked Agda/Lean~4 proofs---and a team chooses its point by a pipeline's
criticality, complexity, and scale. The deployed checker carries the method in
production across our case studies, while a full PDL mechanization in Agda shows the
proof end is reachable for a real pipeline, collapsing the runtime residue to the
boundary alone. And because correctness is
computed from grain, not data, the method is universal over the operation set and
carrier (PDD Universality, \S\ref{sec:background})---PDA over batch collections is
one instance, a windowing algebra over unbounded streams another. And the method is
not even specific to pipelines: a query is a single grain-bearing transformation---the
one-node case---so the same design-time check certifies the AI-generated queries that
run \emph{on top of} the data assets, a correctness layer the ontology/semantic-layer
stack otherwise lacks (\S\ref{sec:beyond}). The engineer's role shifts accordingly: he
\emph{authors and validates the specification}---the data denotation of source and
target and the business rules---and \emph{checks the certificate} the agent ships (a
machine-checked proof, or an evidence ledger), while the algebra and the generated
queries verify the rest. This is a \emph{proof-carrying}~\cite{necula1997pcc} discipline
for AI-generated pipelines: an untrusted producer, a cheaply-checkable certificate, and
a human freed to do the one thing no verifier can---get the specification right. Future work extends the PDL
mechanization from the exemplar to the full operation catalog, instantiates PDD on a
streaming carrier, broadens the pattern library, and realizes the query-layer checker
of \S\ref{sec:beyond} for AI-generated queries over semantic layers.

%% --- Appendix (extended/arXiv version only) --------------------------------
\ifextended
\appendix
\section{Deferred Proofs}
\label{app:proofs}

\begin{proof}[Proof of Theorem~\ref{thm:univ-key-grain}]
Throughout, $k:R\twoheadrightarrow K$ is a field-set projection, so $K$ retains a
subset of $R$'s fields and $k$ is surjective.

\emph{(1)$\Rightarrow$(2):} the elements of any collection $c:\Coll{R}$ are values of
$R$, so a projection injective on the type $R$ is injective on the elements of $c$,
hence a superkey of every $c$.

\emph{(2)$\Rightarrow$(1):} contrapositive. If $k\,r_1=k\,r_2$ with $r_1\neq r_2$, then
$\textit{grain}\,r_1\neq\textit{grain}\,r_2$ by Grain
Uniqueness~\cite{graintheory-pods} (the grain projection is injective), so
$\{r_1,r_2\}$ is a well-formed grain-unique collection of type $C\,R$---and $k$ is not a
superkey of it. Hence $k$ is not a superkey of \emph{every} $c$.

\emph{(1)$\Leftrightarrow$(3):} $k$ is given as a surjection, and a surjection is
injective iff it is a bijection, i.e.\ an isomorphism. So $k$ injective on $R$ (1) is the
same as $k$ being an isomorphism $R\xrightarrow{\sim}K$ \emph{witnessed by $k$ itself}
(3); either way $R\cong K$, hence $R\eqg K$. (This is why clause~(3) must name $k$ as the
witness rather than assert the bare $R\eqg K$: an arbitrary surjection between
grain-equivalent types need not be injective---only $k$'s \emph{being} the isomorphism
gives injectivity.)

\emph{Minimality.} Suppose $k$ satisfies (1)--(3), so $K\cong R$. By definition $K$ is a
\emph{grain} of $R$ precisely when it is moreover \emph{irreducible}: no proper subtype
$K'\psub K$ is isomorphic to $R$ (Grain Inference Theorem~\cite{graintheory-pods},
condition~(iii)). Two observations rewrite this irreducibility as a condition on $k$.
First, since $k$ is a \emph{field-set} projection ($K$ a subset of $R$'s fields), the
proper subtypes $K'\psub K$ are exactly the images of the proper sub-projections
$k':R\twoheadrightarrow K'$ of $k$ (those retaining a proper subset of $K$'s fields):
subtypes of $K$ and sub-projections of $k$ correspond one-to-one. Second, by the
equivalence (1)$\Leftrightarrow$(3) already proved, applied to $k'$, such a $k'$ satisfies
(1)--(3) iff $K'\cong R$. Composing the two, ``some proper subtype $K'\psub K$ is
isomorphic to $R$'' and ``some proper sub-projection $k'$ of $k$ satisfies (1)--(3)'' are
\emph{the same statement}; negating, $K$ is irreducible---hence a grain---iff no proper
sub-projection of $k$ satisfies (1)--(3). The quantification is over
sub-\emph{projections}, a type-level notion, so minimality is type-level, not
per-collection.
\end{proof}

\begin{proof}[Proof of Theorem~\ref{thm:opgen}]
Let $p=\textsf{op}_1\seqc\cdots\seqc\textsf{op}_n$; the DAG case follows by topological
order.

\emph{Denotation totality.} Every type appearing in $p$ is an algebraic data type, and
grain---hence the denotation $\sem{R}=(\grain{R},\ek{R},\BC{R})$---is defined for every
ADT, inductive or coinductive~\cite{graintheory-pods}. So $\sem{\cdot}$ is total on the
types $p$ traverses, independently of the carrier $F$ chosen to hold the data.

\emph{(i)--(ii) Homomorphism and rule composition.} By hypothesis each $\textsf{op}_i$
is faithful: $\sem{\textsf{op}_i\,c}=\calcg[\textsf{op}_i]\,\sem{c}$. We show
$\sem{p\,c}=(\calcg[\textsf{op}_n]\circ\cdots\circ\calcg[\textsf{op}_1])\,\sem{c}$ by
induction on $n$. The base case $n=1$ is faithfulness of $\textsf{op}_1$. For the step,
write $q=\textsf{op}_1\seqc\cdots\seqc\textsf{op}_k$ and
$C_q=\calcg[\textsf{op}_k]\circ\cdots\circ\calcg[\textsf{op}_1]$, so the induction
hypothesis reads $\sem{q\,c}=C_q\,\sem{c}$. Then
\[
\begin{aligned}
\sem{(q\seqc\textsf{op}_{k+1})\,c}
&=\calcg[\textsf{op}_{k+1}]\,\sem{q\,c}\\
&=\calcg[\textsf{op}_{k+1}]\,(C_q\,\sem{c}),
\end{aligned}
\]
by faithfulness of $\textsf{op}_{k+1}$ (on input $q\,c$) then the induction hypothesis. This is the
homomorphism~(i); reading off its grain component is the rule composition~(ii)---the
compositionality of the grain homomorphism (Thm~\ref{thm:grain-hom}).

\emph{(iii) Decidability.} $\calcg[p]$ threads the $n$ per-operation rules through the
DAG by the \calcg{} verification algorithm of Part~I---sequential composition is
function composition of the per-vertex steps, fan-out reuses a computed grain, and
fan-in applies the binary rule, so the three patterns exhaust the
DAG~\cite{graintheory-pods}. Each rule is a finite operation on field sets
($\tun,\tin,\tdiff,\subt$) over the schemas of $p$'s types; the pipeline is a finite DAG
and every schema is a finite field set, so $\calcg[p](\grain{\textsf{Source}})$ is
computed in finitely many steps and compared to $\grain{\textsf{Target}}$ by decidable
field-set equality---no data is read. Part~I bounds this at $O(|V|\cdot k\cdot|F|)$ time
and $O(|V|\cdot|F|)$ space for field-set inference rules, on a DAG of $|V|$ vertices with
fan-in $k$ and at most $|F|$ fields per type~\cite{graintheory-pods}; for a general
grain-inferring $\textsf{Op}$ the cost is $O(|V|)$ times the per-operation rule cost. The
data's cardinality, hence whether $F$ is bounded, never enters: a windowed stream is
verified exactly as a finite table.
\end{proof}

\begin{proof}[Proof of Theorem~\ref{thm:pct}]
Write $\textsf{pip}=\textsf{op}_1\seqc\cdots\seqc\textsf{op}_n$ for a linear chain; the
DAG case is identical taken in topological order. Each layer is discharged on the
type-level data $\sem{\cdot}=(\grain{\cdot},\ek{\cdot},\BC{\cdot})$ and on the
constraint specification, never on a collection.

\emph{Grain.} CalcG composes step by step,
\[\calcg[\textsf{op}_1\seqc\textsf{op}_2]=\calcg[\textsf{op}_2]\circ\calcg[\textsf{op}_1]\]
(\S\ref{sec:correctness:calcg}), so $\calcg[\textsf{pip}]$ is the composite of the
per-operation rules. Hypothesis~(ii) states this composite carries
$\grain{\textsf{Source}}$ to $\grain{\textsf{Target}}$; by Theorem~\ref{thm:calcg} the
output collection's grain equals the target's, for every input, with no data access.

\emph{Behavioral class.} Each operation's BC pre/postcondition is a type parameter, so
(i)---the chain type-checks---makes every BC transition admissible by construction; the
deliberate changes are the re-lenses of \S\ref{sec:correctness:regrain}, sound by
Corollary~\ref{cor:sound-regrain}.

\emph{Domain.} Let $(\textsf{Pre}_i,\textsf{Intra}_i,\textsf{Post}_i)$ be the contract
of $\textsf{op}_i$, with $\textsf{Pre}=\textsf{Pre}_1$ and $\textsf{Post}=\textsf{Post}_n$.
Each operation is correct against its own contract---%
$\textsf{Pre}_i(c)\Rightarrow\textsf{Post}_i(c,\textsf{op}_i\,c)\wedge\textsf{Intra}_i(\dots)$---%
by its denotation (the \emph{(spec)} clause), independently of data. We show by
induction on $i$ that $c\models\textsf{Pre}$ implies the prefix
$\textsf{op}_1\seqc\cdots\seqc\textsf{op}_i$ produces a state satisfying
$\textsf{Post}_i$. \emph{Base:} $\textsf{op}_1$ from $\textsf{Pre}_1=\textsf{Pre}$.
\emph{Step:} the induction hypothesis gives $\textsf{Post}_i$ of the intermediate
state; obligation~(iii), $\textsf{Post}_i\Rightarrow\textsf{Pre}_{i+1}$---discharged at
composition time, or inherited as an axiom when contract-derivable
(\S\ref{sec:correctness:domain})---supplies $\textsf{Pre}_{i+1}$, and the correctness of
$\textsf{op}_{i+1}$ yields $\textsf{Post}_{i+1}$. At $i=n$ this is
$\textsf{pip}\,c\models\textsf{Post}$, the $\textsf{Intra}$ conjuncts accumulated along
the chain. Every implication used is type-level, so the entailment is established
without reading $c$.

\emph{Transfer to the concrete domain.} The argument lives in the semantic domain
$\sem{\cdot}$. The grain homomorphism $\sem{\textsf{pip}\,c}=\sem{\textsf{pip}}\,\sem{c}$
(\S\ref{sec:background}) makes the grain and BC reasoning faithful to the executed
pipeline; each domain predicate is interpreted over the same collections, and its
compiled runtime check returns the empty set iff the predicate holds
(Theorem~\ref{thm:compile}). Hence
$c\models\textsf{Pre}\Rightarrow\textsf{pip}\,c\models\textsf{Post}\wedge\textsf{Intra}$
holds concretely, and the only obligation that consults data is $c\models\textsf{Pre}$.
\end{proof}

\begin{proof}[Proof of Corollary~\ref{cor:sound-regrain}]
The precondition ``$k$ is a superkey of the input $c'$'' says exactly that $c'$ is
grain-unique at $K$ (Thm.~\ref{thm:univ-key-grain}, collection reading), so $c'$ is a
well-formed collection of the re-declared type and the re-lens neither drops an element
nor chooses a representative---it is an identity on data. By hypothesis $t$'s
postcondition makes $k$ a key of $t\,c$ for every input $c$, discharging the
precondition on every collection in the image of $t$; the type-checker verifies the
composition with no premise depending on data. The result denotation is the one
$\texttt{declare-$\langle$bc$\rangle$}$ declares---grain $K$, class $\langle\textsf{bc}\rangle$,
entity key derived from $K$---and is exact when the postcondition gives a minimal key
(Thm.~\ref{thm:univ-key-grain}).
\end{proof}

\fi

%% --- Bibliography ----------------------------------------------------------
\bibliographystyle{ACM-Reference-Format}
\bibliography{bibliography/references}

\end{document}